\documentclass[11pt]{article}

\input{setup.sty}

\title{Geometric Entropies and their Hamiltonian Flows}

\author[1]{Xi Dong,} \emailAdd{xidong@ucsb.edu}

\author[1]{Donald Marolf} \emailAdd{marolf@ucsb.edu}

\author[1,2]{and Pratik Rath} \emailAdd{pratik\_rath@berkeley.edu}

\affiliation[1]{Department of Physics, University of California, Santa Barbara, CA 93106, USA}
\affiliation[2]{Center for Theoretical Physics and Department of Physics,
University of California, Berkeley, CA 94720, USA}

\abstract{
In holographic theories, the Hubeny-Rangamani-Takayanagi (HRT) area operator plays a key role in our understanding of the emergence of semiclassical Einstein-Hilbert gravity. 
When higher derivative corrections are included, the role of the area is instead played by a more general functional known as the geometric entropy. 
It is thus of interest to understand the flow generated by the geometric entropy on the classical phase space. In particular, the fact that the associated flow in Einstein-Hilbert or Jackiw-Teitelboim (JT) gravity induces a relative boost between the left and right entanglement wedges is deeply related to the fact that gravitational dressing promotes the von Neumann algebra of local fields in each wedge to type II.  This relative boost is known as a boundary-condition-preserving (BCP) kink-transformation. In a general theory of gravity (with arbitrary higher-derivative terms),  it is straightforward to show that the flow continues to take the above geometric form  when acting on a spacetime where the HRT surface is the bifurcation surface of a Killing horizon.  However, the form of the flow on other spacetimes is less clear.

In this paper, we use the manifestly-covariant Peierls bracket to explore such flows in two-dimensional theories of JT gravity coupled to matter fields with higher derivative interactions.
The results no longer take a purely geometric form and, instead, demonstrate new features that should be expected of such flows in general higher derivative theories.
We also show how to obtain  the above flows using Poisson brackets.}

\begin{document}

\maketitle

\section{Introduction}
\label{sec:intro}
The Bekenstein-Hawking formula for black hole entropy provides a semiclassical window into quantum gravity \cite{Bekenstein:1973ur} and motivates the idea of holography \cite{tHooft:1993dmi,Susskind:1994vu}, a feature beautifully manifested by the AdS/CFT correspondence \cite{Maldacena_1999}. 
In the setting of AdS/CFT, Ryu and Takayanagi (RT)  \cite{Ryu:2006bv,Ryu:2006ef} argued that a similar formula should describe von Neumann entropies of regions in the CFT.
The covariant Hubeny-Rangamani-Takayanagi (HRT) version of their formula is \cite{Hubeny:2007xt}
\begin{equation}
	S(R) = \frac{\mathcal{A}(\gamma_R)}{4G},
\end{equation}
where $R$ is a region in the boundary CFT and $\gamma_R$ is the minimal, extremal surface homologous to $R$ and anchored to $\partial R$. 
For more general theories of gravity, the functional $\frac{\mathcal{A}(\gamma_R)}{4G}$ in the HRT formula is replaced by the geometric entropy, $\sigma(\gamma_R)$, which includes higher derivative corrections \cite{Wald:1993nt,Dong:2013qoa,Camps:2013zua,Miao:2014nxa}. 

The HRT formula is a remarkable demonstration of how geometric features of the bulk theory are associated with quantum entanglement in the dual CFT, and it has proven to be an invaluable tool in studying such entanglement.  
Although best understood in the context of AdS/CFT, the fact that the HRT formula itself follows from computing the bulk gravitational path integral \cite{Lewkowycz:2013nqa,Dong:2016hjy,Colin-Ellerin:2020mva,Colin-Ellerin:2021jev} suggests that it could apply more generally.\footnote{See Ref.~\cite{Colafranceschi:2023moh} for related results that follow without assuming the existence of a dual CFT. Attempts to apply the HRT formula in spacetimes that are not asymptotically locally AdS include \cite{Sanches:2016sxy,Nomura:2018kji,Dong:2020uxp,Grado-White:2020wlb,Bousso:2022hlz,Bousso:2023sya}.}  
It is thus of great interest to understand the geometric entropy in detail.

The present work seeks to understand the algebraic significance of the geometric entropy by analyzing the Hamiltonian flow generated by the geometric entropy in the classical limit.  In this limit, the geometric entropy $\sigma$ \ is an observable on the gravitational phase space\footnote{This essentially follows from the Lewkowycz-Maldacena derivation \cite{Lewkowycz:2013nqa} of the RT formula \cite{Ryu:2006bv,Ryu:2006ef}, appropriately generalized to the HRT case~\cite{Dong:2016hjy} and higher derivative gravity \cite{Dong:2013qoa,Camps:2013zua,Dong:2017xht,Dong:2019piw}.} and thus defines a Hamiltonian vector field via the usual formula
\begin{equation}\label{eq:flow}
	\frac{dO}{ds}= \{O,\sigma(\gamma_R)\},
\end{equation}
where $O$ is an arbitrary observable and $\{f,g\}$ is the Poisson bracket defined using the symplectic structure of the phase space.  
Here we follow the standard convention for Hamiltonian flow. This agrees with the usual convention for the one-sided modular flow as explained in Ref.~\cite{xi}, but differs in sign from that used in Ref.~\cite{Kaplan:2022orm} as well as the Connes cocycle flow discussed in Ref.~\cite{Bousso:2020yxi}.

We will refer to the flow in \Eqref{eq:flow} as \textit{geometric entropy flow}.
Below, we will find it useful to replace the Poisson bracket with the completely equivalent Peierls bracket \cite{Peierls:1952cb}, which in particular maintains manifest covariance; see \secref{sub:peierls} for a review of this formalism.

For the case of Einstein-Hilbert gravity, it was argued in Ref.~\cite{Bousso:2020yxi} and then shown in Ref.~\cite{Kaplan:2022orm} that the geometric entropy flow takes a simple geometric form known as a boundary-condition-preserving (BCP) kink-transformation.
To describe this flow, consider any Cauchy slice $\Sigma$ that contains the HRT-surface $\gamma_R$. 
All observables are then determined by the initial data ${\cal C}$ on such a slice, i.e., the induced metric $h_{ij}$ and the extrinsic curvature tensor $K_{ij}$. This is of course the case for any two-derivative theory, but it remains true in the presence of higher derivative corrections when such terms are treated perturbatively. 
As shown in \figref{fig:kink}, the intrinsic coordinates on $\Sigma$ can be decomposed into $x_{||}$  (tangential to $\gamma_R$) and $x_{\perp}$ (which measures proper distance normal to $\gamma_R$). The BCP kink-transformation is  then defined by replacing ${\cal C}$ with new Cauchy data ${\cal C}_s$ according to the rule
\begin{align}\label{eq:kink}
        K_{\perp \perp}&\rightarrow K_{\perp \perp} + 2\pi s \,\delta_{\Sigma}\(\gamma_R\),
\end{align}
with all other Cauchy data left unchanged.  
In \Eqref{eq:kink},
$\delta_{\Sigma}\(\gamma_R\)$ is a delta function that at any fixed value of $x_{||}$ satisfies $\int_{\text{fixed }x_{||}} dx_\perp f(x_\perp) \delta_{\Sigma}\(\gamma_R\) = f(0)$.
That an analogous result holds for Jackiw-Teitelboim (JT) gravity with minimally-coupled matter can be seen from the fact that pure JT gravity can be obtained by dimensional reduction of Einstein-Hilbert gravity in the near-horizon limit associated with extremal black holes.

\begin{figure}
\centering
        \includegraphics[width=0.8
        \textwidth]{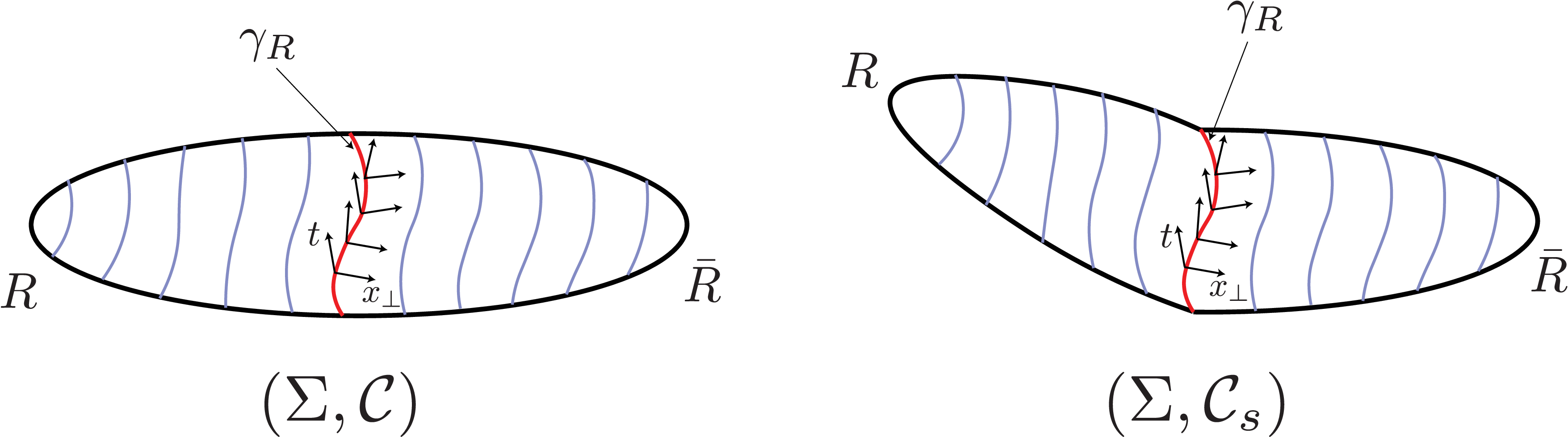}
        \caption{The boundary-condition-preserving (BCP) kink-transformation modifies the Cauchy data ${\cal C}$ on a bulk Cauchy slice $\Sigma$ to new data ${\cal C}_s$ by adding a delta function in the codimension-1 extrinsic curvature component $K_{\perp \perp}$, where as shown $x_{\perp}$ represents the intrinsic coordinate on $\Sigma$ that is orthogonal to the HRT surface $\gamma_R$. The qualifier BCP emphasizes that the slice $\Sigma$ with the new data ${\cal C}_s$ remains glued to the asymptotic boundary regions $R$ and $\bar{R}$ in precisely the same way as for $\Sigma$ with the original data ${\cal C}$.}
        \label{fig:kink}
\end{figure}

As shown in \figref{fig:kink}, this defines new initial data that can be evolved to obtain a new spacetime and to thus compute the effect of this flow on any observable. 
Further, the qualifer BCP is meant to emphasize that the asymptotic boundary conditions are left invariant, i.e., $\Sigma$ with Cauchy data ${\cal C}_s$ remains glued to the boundary in the same way as before the flow. 
As expected from the flow generated by a diffeomorphism-invariant observable, for Einstein-Hilbert gravity it was shown in Ref.~\cite{Bousso:2020yxi} that the BCP kink-transformation preserves the constraint equations and provides valid initial data for evolution.  

It was suggested in Refs.~\cite{Bousso:2020yxi,Kaplan:2022orm} that this simple geometric description could be universally valid in general gravitational theories that include arbitrary perturbative higher derivative terms. 
For a particular higher derivative theory of topologically massive gravity in 2+1 dimensions with a gravitational Chern-Simons term, this was then established in Ref.~\cite{Kaplan:2023oaj}. 

It is also straightforward to show for a general theory of gravity (with arbitrary perturbative higher derivative terms)  that the flow again acts as  a BCP kink transformation on solutions where the HRT surface is the bifurcation surface of a Killing horizon.  We give a general abstract argument here, but an argument by explicit computation is also given in Appendix~\ref{app:sigmaW}.  To begin the abstract argument, note that the Killing field allows one to remove the delta-function in \Eqref{eq:kink} by moving the part of the Cauchy surface on e.g.\ the right side of the HRT surface by a unit Killing parameter along the Killing field while leaving the part on the left side  unchanged.  Doing so leaves us with Cauchy data that is the original $\mathcal{C}$ except for a shift of the coordinate labels along the intersection of our Cauchy surface with the right part of any boundary and, in the presence of Maxwell or Yang-Mills fields, a possible additional overall charge rotation.\footnote{\label{foot:charge} This possibility stems from the fact that the Killing field may preserve charged fields only up to such charge rotations.} This is precisely the shift in Cauchy data that is generated by $\frac{2\pi}{\kappa}(E-\Omega J - \Phi Q)$  where $\Omega$ is the angular velocity of the horizon-generating Killing field, $\Phi$ is the electric potential of the horizon, and $\kappa$ is the surface gravity of the horizon.  Thus,  on such spacetimes, the BCP kink-transformation coincides  with the flow generated by the Hamiltonian vector field of  $\frac{2\pi}{\kappa}(E-\Omega J - \Phi Q)$. 

Let us therefore recall the first law of horizon thermodynamics derived in Refs.~\cite{Wald:1993nt,Iyer:1994ys}, which shows that arbitrary variations in $E-\Omega J - \Phi Q$ about the above spacetimes are precisely $\frac{\kappa}{2\pi}$ times the corresponding variations in $\sigma_W$ where $\sigma_W$ is the (Wald) entropy defined in Ref.~\cite{Wald:1993nt}.  The same proportional relationship must thus hold on the above class of spacetimes between the Hamiltonian vector fields of $\sigma_W$ and $E-\Omega J - \Phi Q$.  Putting this together with the above observations then shows that the BCP kink-transformation coincides on such spacetimes with the flow generated by the Hamiltonian vector field of $\sigma_W$.  Finally, let us also recall that the extrinsic curvature of a Killing-horizon bifurcation surface must vanish along with all of its derivatives, and that the geometric entropy $\sigma$ coincides with $\sigma_W$ up to terms at least quadratic in extrinsic curvatures and their derivatives \cite{Dong:2013qoa}.  Thus all variations of $\sigma$ about the above spacetimes coincide with those of $\sigma_W$.  It follows that their Hamiltonian vector fields agree and, in particular, that $\sigma$ generates the BCP kink-transformation \Eqref{eq:kink} when acting on such spacetimes.

However, in more general circumstances we show below that the flow generated by geometric entropy typically fails to coincide with the simple geometric form \Eqref{eq:kink}.  This failure is associated with the fact that the BCP kink-transformation by itself does not always preserve the constraint equations in higher derivative gravity. We will show this explicitly for the example theories in \secref{sec:hdev}; see, e.g., Eq.~\er{eq:phi_ch} and the comments thereafter.
But for the moment we merely note that any hope that the geometry entropy flow in higher derivative gravity is generally a BCP kink-transformation will be spoiled by field redefinitions. 

In particular, recall that a perturbative field redefinition can be thought of as a coordinate transformation on phase space.
It modifies the perturbative higher-derivative interactions of the theory and maps solutions to solutions, while presenting  the solutions in different variables.
The geometric entropy, being an entanglement entropy in the CFT and a well-defined observable on the gravitational phase space, transforms as a scalar under such field redefinitions \cite{Dong:2023bax}. Thus, the flow generated by it must be covariant. As a result, if we begin with Einstein-Hilbert gravity and apply a field redefinition that mixes $K_{\perp \perp}$ with non-metric degrees of freedom, we necessarily obtain a theory where the flow differs in form from the BCP kink-transformation of \Eqref{eq:kink}. 

In order to better understand the flows generated by general geometric entropies, we will compute results for some example theories that illustrate generic features expected to arise when including perturbative higher-derivative and non-minimal couplings.

In general, as for the BCP kink-transformation, the flow takes a relatively simple form in terms of the transformation of the initial data on a Cauchy slice containing the HRT surface.
We analyze this by computing the Peierls brackets of the geometric entropy with the initial data on such a Cauchy slice.

We start in \secref{sec:ein} by warming up with a simple two-dimensional dilaton theory given by JT gravity coupled to a massless scalar field.\footnote{For  simplicity we work with JT gravity without a cosmological constant.} While the geometric entropy flow in this theory is already well-known to be just the BCP kink-transformation, this example
allows us to illustrate the method being used and to describe the geometry of Lorentzian cones which play an important role in the Peierls bracket calculation.
In \secref{sec:hdev}, we then consider two theories of JT gravity coupled to massless scalar fields with higher-derivative interactions.
Each of these theories has a non-trivial contribution to the geometric entropy coming from these higher-derivative terms. 
In each of these examples, we work to leading order in the higher-derivative coupling and apply our formalism to work out the flow generated by the geometric entropy.
The general feature we find is that the transformation of the initial data involves the introduction of additional singularities in other fields such as the dilaton and matter fields in addition to the singularities in the extrinsic curvature predicted by \Eqref{eq:kink}.
These singularities are localized at the HRT surface as expected.

We then discuss various aspects of our work in \secref{sec:discussion}. These include integrating the linearized flow to finite flow parameter, as well as the role of extremality in our calculation. We also comment on the connection of our results to modular flow in AdS/CFT (on which we will further elaborate in Ref.~\cite{xi}) and describe a new potential way to derive the geometric entropy using Lorentzian methods.

Finally, we provide several appendices to complement our results. We briefly review the covariant phase space formalism in Appendix~\ref{app:covphase}. In Appendix~\ref{app:poisson}, we then repeat our calculations of the geometric entropy flow using the Poisson/Dirac bracket formalism. In Appendix~\ref{app:lorentzian}, we define a quantity called susceptibility in an analogous fashion to the Euclidean derivation of geometric entropy. In Appendix~\ref{app:extra}, we analyze the flow generated by other observables that differ from the geometric entropy at first order in the higher-derivative couplings.


\section{Review}
\label{sec:review}

\subsection{Geometric entropy}
\label{sub:geometric}

The geometric entropy represents a classical contribution to the entropy in gravitational theories. 
An interesting feature of gravity is that the geometric entropy is localized on a codimension-2 surface.
The most familiar such example is the Bekenstein-Hawking entropy of static black holes, $\frac{A}{4G}$, which was computed using the Euclidean gravitational path integral in Ref.~\cite{Gibbons:1976ue}. 
The Euclidean definition of geometric entropy was then significantly extended to include more general states and theories in Refs.~\cite{Lewkowycz:2013nqa,Dong:2013qoa,Camps:2013zua}. 
Although motivated by the calculation of entropy in the dual boundary CFT in the context of AdS/CFT, a corresponding analysis can sometimes be applied in non-AdS contexts as well.  

We now review the definition of geometric entropy, closely following the presentation in Ref.~\cite{Dong:2019piw}.\footnote{We refer the reader to Appendix B of Ref.~\cite{Dong:2019piw} for more details on the variational principle for conical spacetimes and the geometric entropy.}
As argued in Ref.~\cite{Lewkowycz:2013nqa}, the geometric entropy can be computed by studying the response of the gravitational action to the insertion of a Euclidean conical defect. 
To be precise, one first considers Euclidean solutions to the equations of motion with boundary conditions that lead to a well defined variational principle.  
For example, in asymptotically AdS spacetimes, one can specify the induced (conformally-rescaled) metric at the asymptotic boundary. 
It was shown in Ref.~\cite{Dong:2019piw} that this variational principle can be extended to include a  codimension-2 Euclidean conical defect of specified opening angle $2\pi m$ without specifying the  induced geometry or other fields on the defect. 
This extension of the variational principle then allows one to compute the geometric entropy by varying the defect angle parametrized by $m$.\footnote{For simplicity, $m$ will be taken to be irrational for some of the discussion below. In certain situations, one can suitably take limits to rational $m$ as discussed in Ref.~\cite{Dong:2019piw}.}

To do so, one first considers a space of configurations which define the meaning of a conical defect for general theories of gravity.
Near the conical defect, the fields take a form called the ``triple expansion"; see Ref.~\cite{Dong:2019piw} for details.
For such configurations, one defines an action of the form:
\begin{equation}\label{eq:Itilde}
	\tilde{I}[g] = \lim_{\epsilon\rightarrow 0^+}\left( \int_{r\geq \epsilon} d^{d+1}x \sqrt{g}\mathcal{L} + I^{\epsilon}_{\text{CT}}\right),
\end{equation}
where $\mathcal{L}$ is the Lagrangian and the action excludes localized contributions from the defect which, without loss of generality, has been placed at $r=0$.
In the presence of higher derivative corrections, local counterterms at the defect (called $I^{\epsilon}_{\text{CT}}$ in \Eqref{eq:Itilde}) are generally required to cancel power-law divergences in the limit $\epsilon \rightarrow 0$. The action also includes the usual boundary terms at asymptotic boundaries (which are not explicitly shown above).
Such an action defines a well-defined variational principle for spacetimes with a conical defect of specified opening angle $2\pi m$.
With the above definition of the action for conical defect spacetimes in general gravitational theories, the geometric entropy $\sigma$ can be defined as the response of the Euclidean action to the insertion of a Euclidean conical defect, i.e.,
\begin{equation}\label{eq:higherent}
\sigma = -	\frac{d\tilde{I}_m}{dm}\bigg|_{m=1}.
\end{equation}
It is important to note that, in addition to the Wald entropy \cite{Wald:1993nt} terms, the geometric entropy includes the extrinsic curvature corrections described in Refs.~\cite{Dong:2013qoa,Camps:2013zua} although there is currently no simple, explicit formula that works for general higher derivative gravity to all orders in the coupling constants. 
One can also show that the equations of motion fix the location of the defect relative to other features of the geometry, thus defining the HRT surface \cite{Dong:2017xht,Dong:2019piw}.

An important aspect of the Euclidean calculation is that the extrinsic curvature terms in the geometric entropy \cite{Dong:2013qoa} arise in a qualitatively different manner from the Wald entropy terms.
When perturbing the opening angle from $m=1$ to $m=1+\delta m$, one can expand the metric as $g_{1+\delta m} = g +\delta g$. 
In this calculation of geometric entropy, the Wald entropy terms arise from contributions to the Lagrangian that are manifestly first order in $\delta g$, whereas the extrinsic curvature terms arise from certain terms in the Lagrangian that are formally second order in $\delta g$.
This is  because such second-order terms integrate to give contributions to the action that are only first order in $\delta m$, a fact which is signaled by the observation that such integrals have logarithmic divergences if not treated carefully \cite{Dong:2013qoa}.
This feature arises from the structure of the triple expansion near the defect.
In \secref{sec:discussion} and Appendix~\ref{app:lorentzian}, we contrast this with a potential Lorentzian derivation of geometric entropy which only relies on formally first-order terms in $\delta g$ due to a simpler ``double expansion" structure in Lorentzian signature.


\subsection{Peierls bracket}
\label{sub:peierls}

Classical mechanics textbooks typically define the Hamiltonian flow on a given phase space in terms of the Poisson bracket.  
However, the fact that the Poisson bracket breaks manifest covariance can render it cumbersome for computations in diffeomorphism-invariant theories.
Furthermore, in higher derivative theories an additional complication arises from the fact that the computation of canonical momenta requires an explicit reduction-of-order procedure\footnote{See Appendix~\ref{app:poisson} where such a procedure is carried out explicitly.}.  

These features motivate us to seek a more convenient formalism in which to carry out our analysis.
To preserve manifest covariance, we work in the covariant phase space formalism which is reviewed in Appendix~\ref{app:covphase}.
In this formalism, the concept of the Peierls bracket is well-adapted to our needs and is completely equivalent to the standard Poisson bracket. 
In general, the Peierls bracket $\{f,g\}$ measures the linear response of observable $f$ under a deformation of the theory defined by the observable $g$.   
One can show that the Peierls bracket is antisymmetric so that  (up to a minus sign) one may alternately compute the linear response of $g$ under a deformation of the theory by $f$. 
The calculation of a Peierls bracket (and its association with linear response) is similar in spirit to that of a Wightman propagator in quantum field theory.  
Indeed, this relation is manifest in quantum field theory through the well-known relation that the commutator function is twice the imaginary part of the Wightman propagator.   That the Peierls bracket describes the classical limit of the quantum commutator also follows directly from the less well-known Schwinger variational principle \cite{Schwinger:1951ex,DeWitt:1964mxt,DeWitt:1984ojp,DeWitt:2003pm}.

\begin{figure}
\centering
        \includegraphics[width=0.8
        \textwidth]{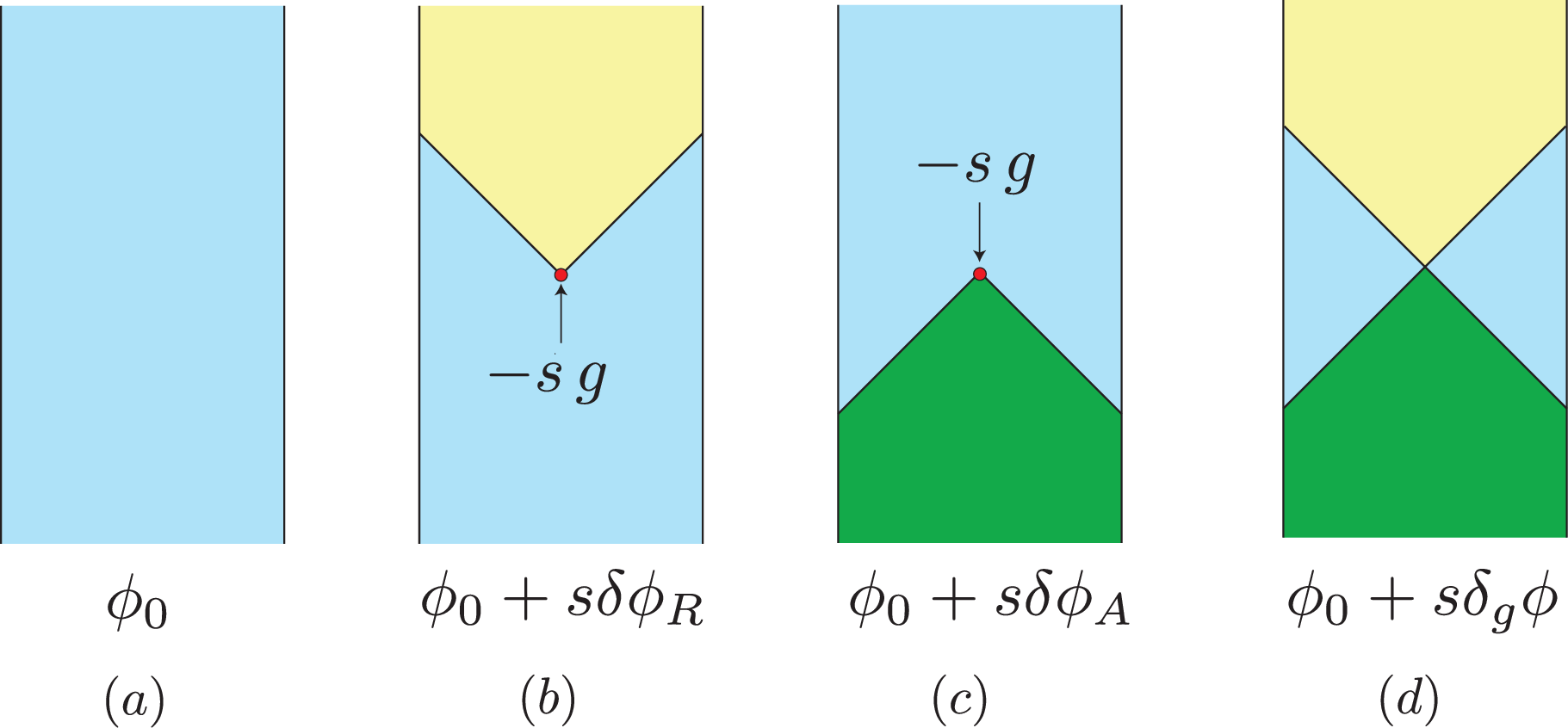}
        \caption{The computation of Peierls bracket $\{f,g\}$ involves four solutions, depicted here for a codimension-2 operator $g$: a) The original solution $\phi_0$ of the undeformed action, b) the retarded solution $\phi_0+s\,\delta\phi_R$ with a source $-s\,g$ (red), c) the advanced solution $\phi_0 + s\,\delta\phi_A$  with the same source, d) the linearized solution $\phi_0+s\,\delta_g \phi$ to the undeformed action, where $\d_g \phi = \d \phi_R-\d\phi_A$.}
        \label{fig:peierls}
\end{figure}

We now describe the computation of the Peierls bracket in more detail.
Consider a theory defined by an action functional $I_0$, and consider two observables $f$ and $g$ on the classical phase space (i.e., the space of classical solutions quotiented by zero modes) of this theory.
In order to compute the Peierls bracket $\{f,g\}$, the procedure is as follows (see \figref{fig:peierls} for an illustration):
\begin{itemize}
		\item Consider an arbitrary configuration $\phi_0$ that solves the equations of motion given by the action $I_0$.
        \item Now, deform the action to $I_{s}\equiv I_0 - s\,g$.
        \item Find a new configuration that solves the linearized (in $s$) equations of motion for the deformed action $I_s$ with a retarded boundary condition, i.e., such that the new configuration $\phi_0+s\, \delta\phi_R$, called the retarded solution, is identical to $\phi_0$ in the past of the source $g$.
        \item Similarly, find an advanced solution $\phi_0+s\, \delta\phi_A$ to the linearized equations of motion of $I_s$ which is identical to $\phi_0$ in the future of the source $g$.
		\item Consider the difference $\delta_g \phi=\delta\phi_R-\delta\phi_A$. 
        This is a solution to the linearized equations of motion of the original, undeformed action $I_0$.
        \item With all these ingredients in hand, the Peierls bracket is then given by
        \begin{equation}\label{eq:flowp}
                \{f,g\}[\phi_0]\equiv \frac{d}{ds}f[\phi_0+s\,\delta_g \phi]\Big\vert_{s=0}.
        \end{equation}
\end{itemize}

Note in particular that the Peierls bracket computation is manifestly covariant and does not require one to define either particular coordinates on phase space or a particular foliation of spacetime. Also note that the procedure above needs an extension of $g$ from the phase space to the off-shell configuration space (so we can use it to deform the action), but  one can show that the final Peierls bracket \er{eq:flowp} is independent of such extensions and is thus umambiguously defined; see e.g.\ Eq. (2.10) of Ref.~\cite{Marolf:1993af}.

Our goal in this paper will be to use this formalism  to understand the flow  
\begin{equation}\label{eq:geomflow}
		\frac{dO}{ds}= \{O,\sg\}
\end{equation}
defined by taking Peierls brackets with $\sigma$. 
Here $O$ is an arbitrary diffeomorphism invariant observable, and we remind the reader that $\sg$ is implictly evaluated at the HRT surface $\gamma$. 
At the linearized level, the addition of $\sg$ into the action generates a distributional source localized on $\gamma$.
In order to find retarded/advanced solutions, we can then solve the equations of motion together with an appropriate matching condition across the distributional source.
This can be done from first principles on a case-by-case basis for different higher derivative theories. 

\Eqref{eq:flowp} makes it clear that the linearized solution $\d_g \p$ precisely describes the change in the original solution under the geometric entropy flow \Eqref{eq:geomflow}. It is manifestly covariant, since it directly describes the change in the entire spacetime all at once without picking a foliation. Thus, for our cases of interest, we will simply compute $\d_g \p$ from which we can then read off the change in any particular observable $O$ . In the simple models that we will study, we will find this linearized solution in closed form.  However, in general it is easier to proceed in analogy with our treatment of the BCP kink-transformation above in which one explicitly describes only the change in the initial data on a slice passing through the HRT surface, from which other changes are then determined by the equations of motion.

\section{Warm-up: JT Gravity with Minimally Coupled Scalar}
\label{sec:ein}

We start off by considering JT gravity with zero cosmological constant coupled to a massless scalar $\y$.\footnote{The cosmological constant has been set to zero for convenience, and adding it in will not significantly modify our conclusions.} The action (ignoring boundary terms) is
\begin{equation}\label{eq:action}
	I_0 = \frac{1}{2} \int d^2 x \sqrt{-g}\[\phi R-(\na \y)^2 \],
\end{equation}
where $\p$ is the dilaton.
The geometric entropy in this model is well known to be $\sg=2\pi \p$ {evaluated at its extremal point. 

The equations of motion of the above theory obtained by varying $\p$ and $\y$ are given by:
\begin{align}
    R&=0,\\
    \na_\m \na^\m \y&=0.
\end{align}
flat, two-dimensional Minkowski space, which we can parametrize using lightcone coordinates which yield the line element
\begin{align}
    ds^2=du\,dv.
\end{align}
In these coordinates, it is well known that a general solution to the free massless scalar field equation is simply given by
\begin{equation}\label{eq:sumLR}
    \psi(u,v) = \psi_R(u)+\psi_L(v).
\end{equation}
The equations of motion from varying the metric are given by:
\begin{equation}\label{eq:dilEOM}
    g_{\m\n} \[-\fr{1}{2} \na^2\p -\fr{1}{4}(\na\y)^2 \] +\fr{1}{2} \na_\m\na_\n\p +\fr{1}{2} \na_\m\y \na_\n\y =0,
\end{equation}
whose $(uv)$-component tells us that the dilaton is also a sum over left and right movers, i.e.,
\begin{equation}\la{eq:phisumLR}
    \phi(u,v) = \phi_R(u)+\phi_L(v).
\end{equation}
The remaining components of \Eqref{eq:dilEOM} then impose constraints on $\phi_R,\phi_L$.

In order to understand the geometric entropy flow, we can compute the Peierls bracket of an arbitrary observable with $\sg$ by finding the linearized solution generated by $\sg$. The first step is to work out the equations of motion for the modified action $I_s=I_0-s \sg$. In principle, field variations may change the coordinate location $x^\m_\g$ of the extremal surface $\g$. This change $\d x^\m_\g$ contributes to the variation of the action by the amount
\be
\label{eq:dImovegamma}
\d I_s = -s \fr{\pa \sg}{\pa x^\m_\g} \d x^\m_\g.
\ee
However, the extremality condition $\pa \sg/\pa x^\m_\g =0$ requires \Eqref{eq:dImovegamma} to vanish.  As a result, we obtain the same variation of the action if we simply fix $\gamma$ to lie at the origin of our coordinates.

With this understanding, varying the dilaton gives
\begin{equation}\label{eq:ric1}
    R= 4\pi s \d_\g,
\end{equation}
where $\d_\g$ is a codimension-2 delta function defined to satisfy $\int d^2 x\sqrt{-g} f(x) \d_\g=f(0)$ for arbitrary test functions $f$.  \Eqref{eq:ric1} completely fixes the geometry to be flat space with a Lorentzian conical defect at $\g$.

\begin{figure}
\centering
        \includegraphics[width=0.7
        \textwidth]{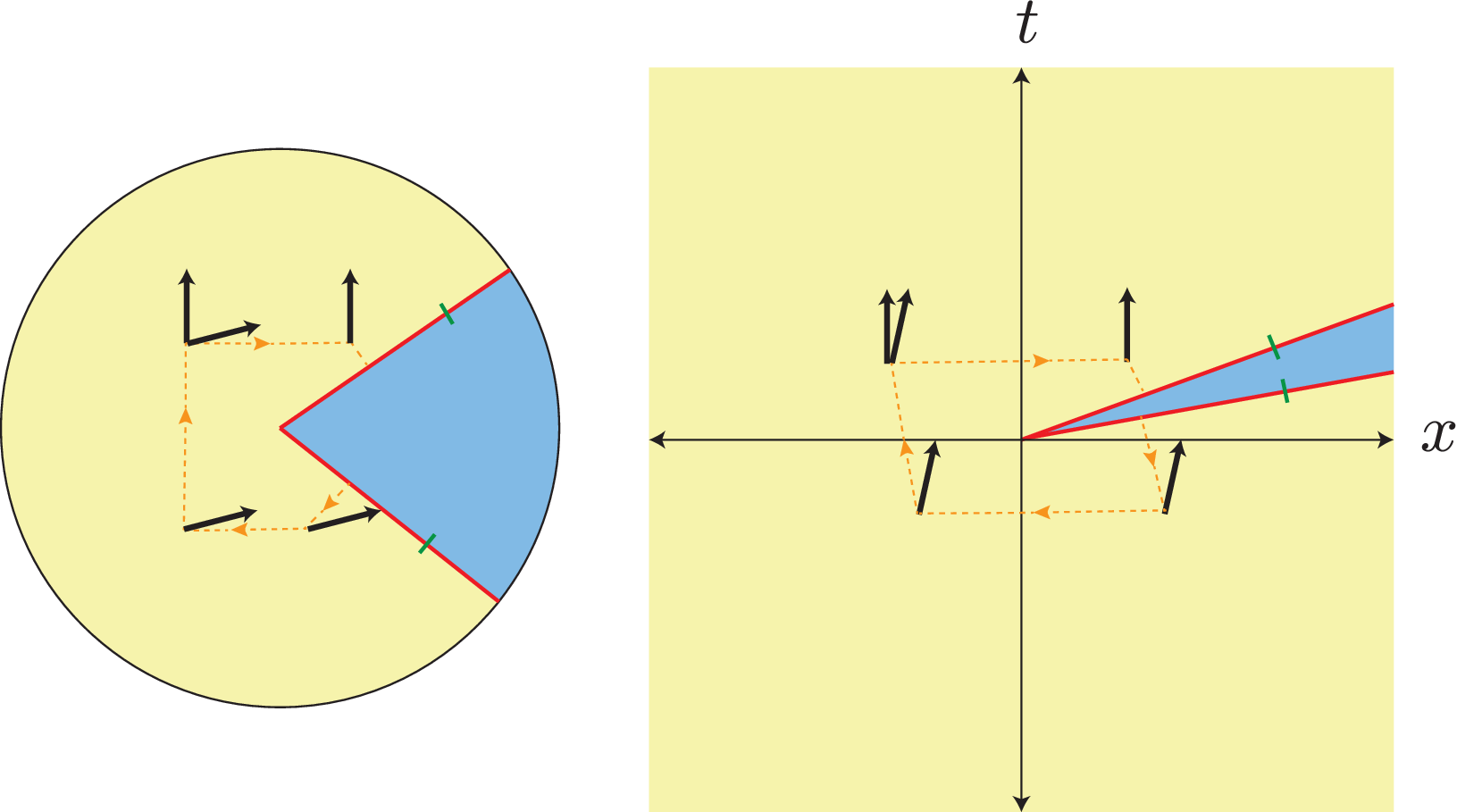}
        \caption{(left): A Euclidean cone is obtained by starting with flat Euclidean space (yellow), identifying two radial slices (red), and deleting the region between them (blue). The space is manifestly flat everywhere except for a curvature delta function at the origin, which can be sensed by the holonomy of a vector parallel-transported around the origin in a loop (dashed orange). (right): Analogous construction of a Lorentzian conical defect is obtained by deleting the region (blue) between two Rindler slices (red) in Minkowski space (yellow) and identifying the slices together. Again, the holonomy of a vector parallel-transported around the origin in a loop (dashed orange) is depicted.}
        \label{fig:cone}
\end{figure}

Since it has not been frequently discussed in the recent literature, it is useful to first review the geometry of a Lorentzian conical defect. 
To set the stage, recall that a Euclidean cone can be constructed from flat (say, two-dimensional) Euclidean space by removing a wedge and identifying the open surfaces as shown in \figref{fig:cone}. 
This identification is allowed due to the $U(1)$ symmetry that relates the induced metrics on the two constant angle slices. 
The signature property of a cone is that a vector parallel-transported around the tip returns with a relative rotation, a holonomy that measures the relevant deficit angle.

In an analogous fashion, one may construct a Lorentzian conical defect by removing a wedge from Minkowski space and identifying the two slices as shown in \figref{fig:cone}.
This again can be done due to the boost symmetry of the original Minkowski space.
One then finds a similar holonomy when one encircles the defect, where the parallel-transported vector returns with a relative boost that measures the strength of the defect.
In fact, in an arbitrary spacetime with a Lorentzian conical defect, we can always zoom in sufficiently close to $\g$ so that it looks like the above example.

There are various charts one could use to describe the geometry on the right side of \figref{fig:cone}. For our calculation here, we find it convenient to work in conformal gauge:
\be
ds^2 = e^{2\w} du dv.
\ee
In this gauge, the Ricci scalar is given by
\be\label{eq:ricc}
\sqrt{-g} R = -4 \pa_{u} \pa_{v} \w.
\ee
Then, by choosing
\begin{equation}
    \w=\pi s\, \q\(-u\)\q\(v\),
\end{equation} where $\q(x)$ is the Heaviside step function, we obtain a chart for the Lorentzian cone that smoothly covers the entire spacetime except the origin where the curvature is singular as expected. This is a convenient choice for the retarded solution, since we have $\w=\pi s$ in the future wedge and $\w=0$ elsewhere as shown in \figref{fig:lcone}. Using \Eqref{eq:ricc}, we find that this solves the equation of motion \Eqref{eq:ric1}.

\begin{figure}
\centering
        \includegraphics[width=0.6
        \textwidth]{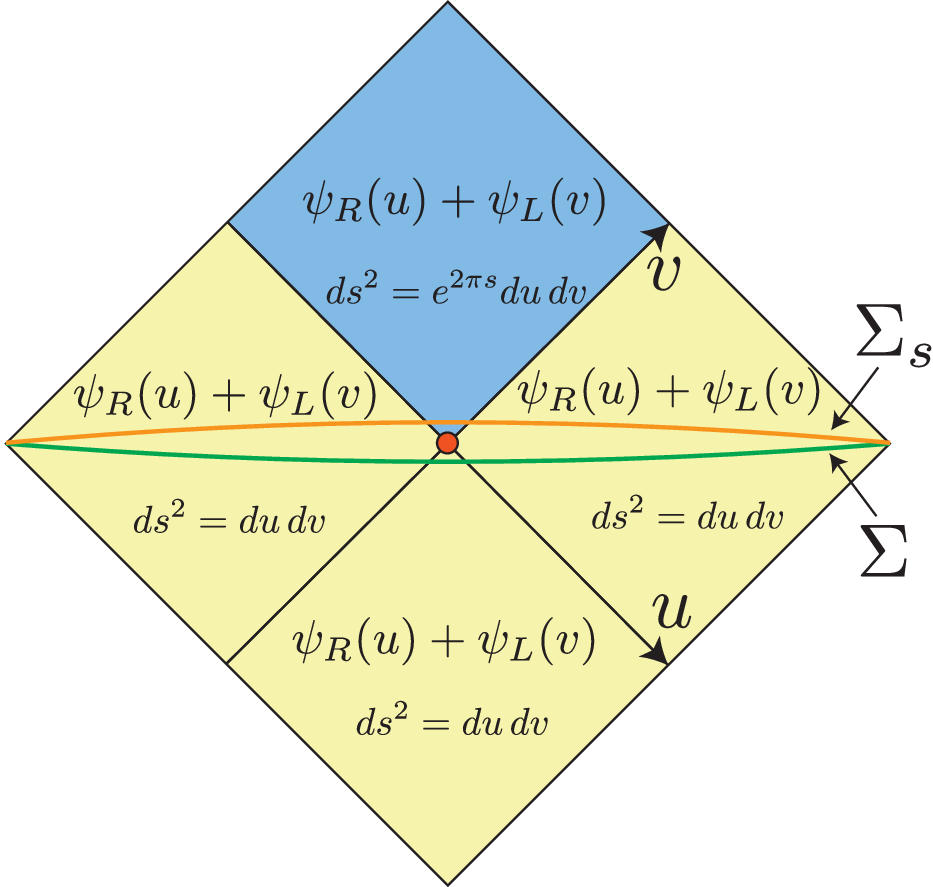}
        \caption{The retarded solution for the massless scalar field $\y$ expressed in terms of a chart that covers the entire spacetime, except the origin (red). The expression for $\y$ in terms of the $u,v$ coordinates shown is independent of $s$. Indeed, the same is true for the dilaton $\p$.  Only the metric conformal factor depends on $s$, and only in the future of $\gamma$ (blue).  The transformation of initial data can be computed by comparing the data on Cauchy slice $\Sigma$ (green) with $\Sigma_s$ (orange).}
        \label{fig:lcone}
\end{figure}

We can then look at the scalar field equation on this conical background. The equation of motion is given by
\begin{equation}
    \na_\m \na^\m \y =2e^{-2\w}\pa_u\pa_v \y = 0.
\end{equation}
This equation is unmodified, as should be expected from the conformal invariance of a massless scalar.
Thus, we simply have the usual massless scalar field solution which is a sum over left and right movers as in \Eqref{eq:sumLR}. It is then straightforward to write down the full retarded solution as shown in \figref{fig:lcone}.

Finally, we can look at the equations of motion obtained by varying the metric, which are the same as \Eqref{eq:dilEOM}, except that the background metric has changed. It is now simplest to work in components. In conformal gauge, the only nonzero Christoffel symbols are
\be\label{eq:christ}
\G^u_{uu} = 2\pa_u \w,\qu
\G^v_{vv} = 2\pa_v \w,
\ee
and therefore the $(uv)$-component becomes
\be
\na_u \na_v \p = \pa_u \pa_v \p = 0
\ee
resulting in a solution of precisely the same form as $\y$ shown in \figref{fig:lcone}. 

The $(uu)$-component then yields the constraint
\be\label{eq:phiuu}
\na_u \na_u \p +(\na_u \y)^2= \pa_u^2 \p - 2\pa_u \w \pa_u \p+(\pa_u \y)^2=0.
\ee
While there is a potential $\d$ function contribution from he term $\pa_u \w=-\pi s\, \d(u)\q(v)$, this contribution vanishes because the extremality condition $\pa_u \phi\vert_{\g}=0$, combined with $\p$ being a sum over left and right movers as shown in \Eqref{eq:phisumLR}, leads to $\pa_u \phi\vert_{u=0}=0$. As a result, \Eqref{eq:phiuu} takes precisely the same form as for $s=0$ and the constraint remains satisfied. A similar argument holds for the $(vv)$-component. Thus, we find that the solution described by \figref{fig:lcone} satisfies all the modified equations of motion everywhere including the origin.

\begin{figure}
\centering
        \includegraphics[width=
        \textwidth]{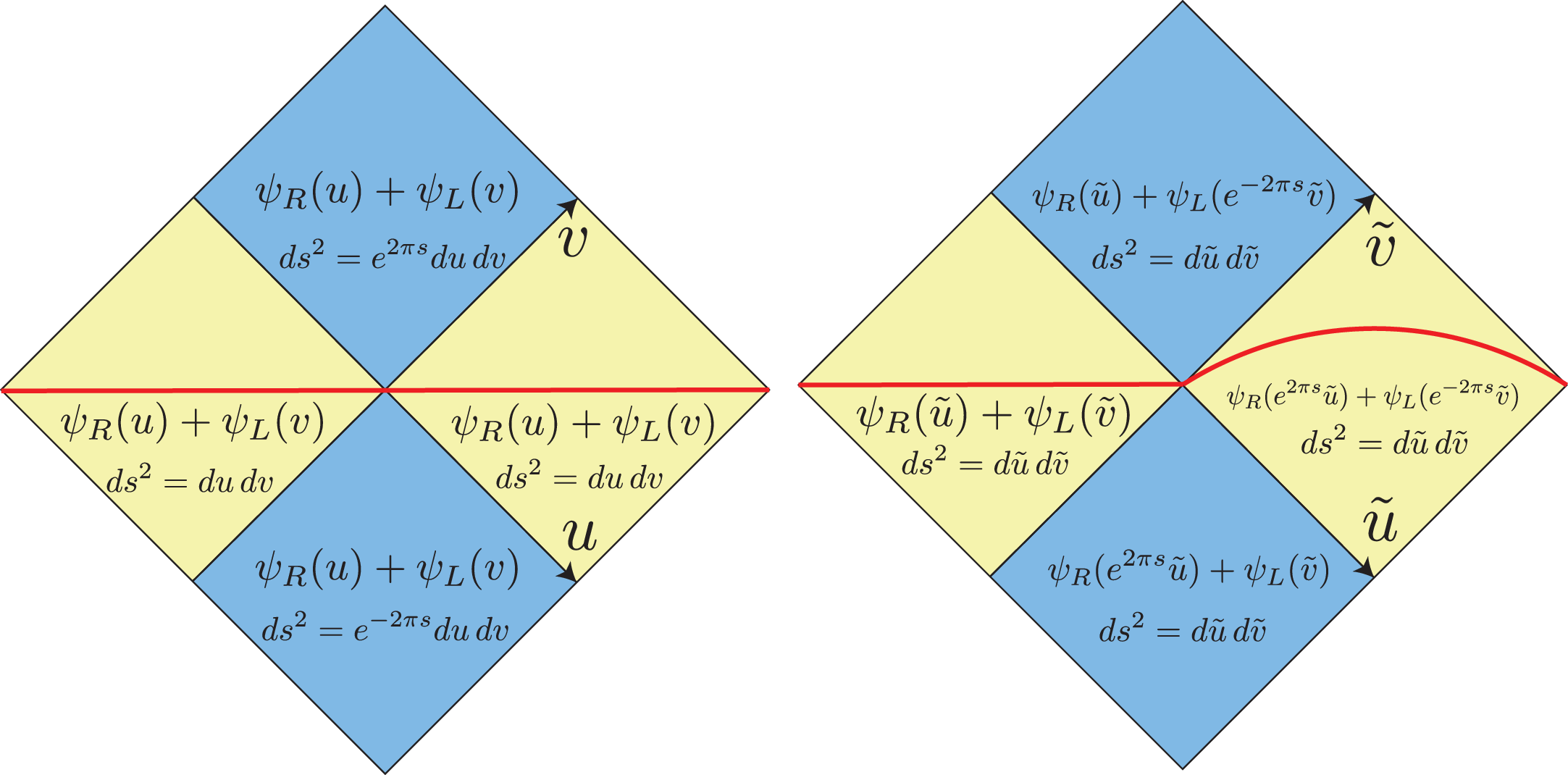}
        \caption{The flowed solution for the massless scalar field $\y$ expressed in terms of coordinates $(u,v)$ (left) and $(\tilde{u},\tilde{v})$ (right). The solution for the dilaton $\p$ also takes the same form. This solution can be obtained by evolving the initial data on the red Cauchy slice, which is obtained by a BCP kink-transformation from the original solution.}
        \label{fig:flowed}
\end{figure}

The advanced solution with  source $-s\sg$  can be constructed analogously using $\w=\pi s \q\(u\)\q\(-v\)$, for which the change is supported purely in the past of $\g$. We can then compute the linearized change in our solution under the flow by taking the difference between the retarded and advanced solutions. With the above choice of coordinates, we find the solution shown in \figref{fig:flowed} (left).  Although we are working at the linearized level, we have nevertheless exponentiated the flow parameter for convenience.

We can also make a coordinate transformation to new coordinates $\(\tilde{u},\tilde{v}\)$ where $\tilde{u}=u$ in the left and future wedges, while $\tilde{u}=u e^{-2\pi s}$ in the past and right wedges. Similarly, we have $\tilde{v}=v$ in the left and past wedges, while $\tilde{v}=v e^{2\pi s}$ in the future and right wedges. In this new choice of coordinates, the metric everywhere takes the form
\begin{equation}
    ds^2=d\tilde{u}\,d\tilde{v},
\end{equation}
while the scalar field and dilaton solutions are as depicted in \figref{fig:flowed} (right).

This choice of coordinates $\(\tilde{u},\tilde{v}\)$ makes it manifest that the flowed solution has a flat metric everywhere, while the matter and dilaton fields have non-analyticities at the Rindler horizons. In particular, the fields are continuous, but their first derivatives have a jump across the horizon. These distributional singularities are of the same flavour as the shocks in the Weyl tensor found in Ref.~\cite{Bousso:2020yxi}.

While we have obtained the full flowed solution in this simple model, it is also useful to describe the change in initial data on a slice passing through the origin. In a general theory, one would do this by deriving the matching conditions across the source in the retarded solution by integrating the equations of motion across the source. This is in fact the approach taken by Peierls in his original work \cite{Peierls:1952cb}.

In this example, since we have the full solution, it is instead easy to read off the transformation by looking at the retarded solution and comparing the data on slice $\Sigma$ just to the past of the origin with the data on slice $\Sigma_s$ just to the future of the origin as shown in \figref{fig:lcone}. We can use the usual Minkowski coordinates defined by $u=x-t$ and $v=x+t$, to define $\Sigma$ as $t=0^-$ and $\Sigma_s$ as $t=0^+$.

From \figref{fig:lcone}, it is straightforward to see that both the matter field and dilaton initial data are completely unchanged. On the other hand, while the induced metric  on $\Sigma_s$ is unchanged, the extrinsic curvature becomes
\begin{align}\label{eq:extchange}
    K_{xx} = K_{uu}+K_{vv}= \frac{\Gamma^v_{vv}-\Gamma^u_{uu}}{2}= 2\pi s\d(x) = 2\pi s \d_\S(\g),
\end{align}
where we have used \Eqref{eq:christ}, and $\d_\S(\g)$ is the delta function defined below \Eqref{eq:kink}, with domain $\S$ and support $\g$. 
The resulting transformation is thus precisely the BCP kink transformation \Eqref{eq:kink}. 
Although we chose to begin with a flat slice $\Sigma$, an analogous analysis for general $\Sigma$ of course continues to agree with \Eqref{eq:kink}.
As shown in \figref{fig:flowed} (right), this result can also be derived by studying $\Sigma$ in terms of the natural light-cone coordinates $\tilde u, \tilde v$ of the transformed metric.

While the Peierls formalism is naturally phrased in terms of linearized solutions, in this case we can integrate the flow to obtain the full solution at finite $s$. The result is identical to \figref{fig:flowed} with $s$ now treated as a finite parameter. 

\section{Examples: JT Gravity with Higher Derivative Interactions}
\label{sec:hdev}
We will now use the above formalism to analyze the flow in theories with higher-derivative interactions. The theories we will consider for illustrative purposes are JT gravity coupled to matter fields with higher-derivative interactions. These theories were discussed in Appendix~A of Ref.~\cite{Dong:2017xht}, where their geometric entropy was also computed using Euclidean methods.

\subsection{Theory 1}
\label{sub:t1}

The first model we consider has an action given by
\begin{equation}\label{eq:action1}
	I_0 = \frac{1}{2} \int d^2 x \sqrt{-g}\[\phi R-(\na \y)^2 +\l \na_\m \na_\n \y \na^\m \na^\n \y \],
\end{equation}
where $\l$ is treated as a perturbatively small coupling. This means that all fields will be expanded in a formal power series in $\l$ and the equations of motion will be solved order by order in $\l$. This ensures that we are not introducing new degrees of freedom into the theory and avoid issues such as the Ostrogradsky instability \cite{Ostrogradsky:1850fid}. In this model, all our results will be $\l$-exact and should be interpreted as a perturbation series to all orders in $\l$.

As before, we can first analyze the general solution to the equations of motion, the simplest being $R=0$. This gives us Minkowski space which we again chart using lightcone coordinates $(u,v)$ such that $ds^2=du\,dv$. The equation of motion for $\psi$ is given by
\begin{equation}\label{eq:psi_lor}
\na_\m \na^\m \psi +\l \na_\m \na^\m \na_\n \na^\n \psi=0,
\end{equation}
where we have commuted covariant derivatives using the fact that the metric is flat.
Treating $\l$ perturbatively, we first look at the leading order solution which is given by
\begin{equation}\label{eq:psi_sol}
    \psi(u,v) = \psi_R(u) + \psi_L(v).
\end{equation}
It is easy to then check that the solution preserves the form of \Eqref{eq:psi_sol} to all orders in $\l$.

Finally, by varying the metric, we obtain the other equations of motion \cite{Dong:2017xht}:
\bm
g_{\m\n} \[-\fr{1}{2} \na^2\p -\fr{1}{4}(\na\y)^2 +\fr{\l}{4} \na_\r\na_\s\y \na^\r\na^\s\y\] +\fr{1}{2} \na_\m\na_\n\p +\fr{1}{2} \na_\m\y \na_\n\y \\
+\fr{\l}{2} \(\na_\m\y \na^2 \na_\n\y +\na_\n\y \na^2 \na_\m\y -\na^2\y \na_\m\na_\n\y -\na_\r\y \na^\r\na_\m\na_\n\y\) = 0.
\em
By using the solution for $\psi$, i.e., \Eqref{eq:psi_sol}, the above equations of motion for the dilaton can then be rewritten in the form
\begin{align}\label{eq:phi_lor_uu}
    \pa_u^2 \tilde{\phi}+\(\pa_u \psi\)^2=0,\\
    \label{eq:phi_lor_vv} \pa_v^2 \tilde{\phi}+\(\pa_v \psi\)^2=0,\\
    \label{eq:phi_lor_uv} \pa_u \pa_v \tilde{\phi}=0,
\end{align}
where we apply the convenient field redefinition $\tilde{\phi} = \phi - \frac{\l}{2} \na_\mu \psi \na^\mu \psi$.
From \Eqref{eq:phi_lor_uv}, it is clear that $\tilde{\phi}$ is also given to all orders in $\l$ by a sum over left and right movers:
\begin{equation}\la{eq:phitsumLR}
    \tilde{\phi}(u,v) = \tilde{\phi}_R(u) + \tilde{\phi}_L(v).
\end{equation}

From the Euclidean analysis, the geometric entropy is given to all orders in $\l$ by \cite{Dong:2017xht}
\begin{equation}
    \sg = 2\pi \phi - \pi \l \(\nabla \y\)^2 = 2\pi \tilde{\p},
\end{equation}
evaluated at its extremal point. We will now compute the flow generated by $\s$. In particular, we will  focus on the Peierls bracket of $\sg$ with initial data on a Cauchy slice passing through the HRT surface $\g$.

The first step is to work out the equations of motion for the modified action $I_s=I_0-s \sg$. From varying the dilaton, we have
\begin{equation}
    R= 4\pi s \d_\g,
\end{equation}
where we have again used the extremality of $\g$ to ignore terms coming from varying the location of the surface. As before, we note that in two dimensions the Ricci scalar completely determines the geometry to be the flat metric with a Lorentzian conical defect at $\g$ (which we place at the origin). Again, we work in conformal gauge:
\be
ds^2 = e^{2\w} du dv,
\ee
with $\w = \pi s\, \q\(-u\)\q\(v\)$ for the retarded solution. Moreover, the Ricci tensor is given by
\begin{equation}\label{eq:ric}
    R_{\a\b} = 2 \pi s \,g_{\a\b} \,\d_\g.
\end{equation}

Varying the matter field $\y$, we obtain the following modified equation of motion (EOM):
\begin{equation}\label{eq:psiret}
    \na_\m \na^\m \y +\l \na_\n \na_\m \na^\m\na^\n \y =  2\pi \l s \na_\m\(\na^\m \y \d_\g\).
\end{equation}
Focusing on the second term, we can use the definition of the Riemann tensor as commutator of covariant derivatives, i.e.,
\begin{equation}\label{eq:riemm}
    \na_\r \na_\m\na_\n \y = \na_\m\na_\n \na_\r \y +R_{\r\m\n\s} \na^\s \y,
\end{equation}
to obtain the identity
\begin{equation}\label{eq:swap}
    \[\na_\m \na^\m,\na_\n\] \y = R_{\a \n}\na^\a \y.
\end{equation}
Using this, we can swap the order of derivatives to obtain 
\begin{equation}\label{eq:comm}
    \na_\n \na_\m \na^\n\na^\m \y= \na_\n\na^\n \na_\m\na^\m \y  +\na^{\a}\( R_{\a\b} \na^\b \y\).
\end{equation}

Using \Eqref{eq:ric}, we find that the last term in \Eqref{eq:comm} precisely cancels the source on the right-hand-side (RHS) of \Eqref{eq:psiret}.
Thus, the EOM simplifies to the form of \Eqref{eq:psi_lor} and we find the same solution \Eqref{eq:psi_sol} satisfying the EOM without any additional singularities.
In summary, the retarded solution for $\psi$ is once again precisely as described in \figref{fig:lcone}.

We then consider the EOMs arising from varying the metric. Using \Eqref{eq:swap} and the $\y$ EOM \Eqref{eq:psi_lor}, we simplify the EOMs to
\bm\label{eq:geom}
g_{\m\n} \[-\fr{1}{2} \na^2\p -\fr{1}{4}(\na\y)^2 +\fr{\l}{4} \na_\r\na_\s\y \na^\r\na^\s\y\] +\fr{1}{2} \na_\m\na_\n\p +\fr{1}{2} \na_\m\y \na_\n\y \\
+\fr{\l}{2} \(R_{\n\a} \na_\m\y \na^\a\y +R_{\m\a} \na_\n\y \na^\a\y -\na_\r\y \na^\r\na_\m\na_\n\y\) = \pi \l s \d_\g \na_\m\y \na_\n\y.
\em
We will now need to work in components and use the Christoffel symbols \Eqref{eq:christ}.

After using \Eqref{eq:riemm}, the $(uu)$-component of \Eqref{eq:geom} simplifies to
\be
\fr{1}{2} \na_u \na_u\p +\fr{1}{2} (\na_u\y)^2 +\fr{\l}{2} \[2R_{uv} \na_u\y \na^v\y -\na^u\y \na_u \na_u \na_u\y -R_{vuuv} (\na^v\y)^2\]
=\pi \l s \d_\g (\na_u\y)^2.
\ee
In the current coordinates, the non-zero components of the Ricci and Riemann tensor are
\be
R_{uv} = \fr{1}{4} e^{2\w} R,\qqu
R_{uvvu} = R_{vuuv} = \fr{1}{8} e^{4\w} R.
\ee
Further using $R=4\pi s \d_\g$, we find
\be
\na_u \na_u\p +(\na_u\y)^2 -\l \na^u\y \na_u \na_u \na_u\y =0.
\ee
Recalling our convenient field redefinition:
\be
\td\p = \p - \fr{\l}{2} \na_\a \y \na^\a \y = \p - \l \na_u \y \na^u \y,
\ee
we find that the $uu$ component simplifies to
\be\label{uueom}
\na_u \na_u\td\p +(\na_u\y)^2 =0.
\ee
In direct analogy with our treatment of  \Eqref{eq:phiuu}, we can now use the extremality condition and \Eqref{eq:phitsumLR}  to reduce \Eqref{uueom} to the unsourced EOM \Eqref{eq:phi_lor_uu}, i.e., to
\begin{equation}\la{uueom2}
    \pa_u^2 \tilde{\phi}+\(\pa_u \psi\)^2=0.
\end{equation}
By symmetry, it is clear that this can be done for the $(vv)$-component as well.

\be
\pa_u \pa_v \td\p = 0,
\ee
Furthermore, using $R=4\pi s \d_\g$ the $(uv)$-component of \Eqref{eq:geom} becomes
\be
\pa_u \pa_v \td\p = 0,
\ee
which agrees with the corresponding unsourced EOM \Eqref{eq:phi_lor_uv}.
Since all three components  agree with their unsourced counterparts, it is clear that the retarded solution for $\td\p$ is the same as the original solution:
\be
\td\p = \td\p_R(u) + \td\p_L(v),
\ee
where we note again that this holds to all orders in $\l$. Thus, the full retarded solution written in terms of $\w$, $\y$, and $\td\p$ again takes the form shown in \figref{fig:lcone}; i.e., the retarded solution in conformal gauge is
\be\label{eq:retnew}
\w = \pi s \q(-u) \q(v), \qquad \y = \y\big|_{s=0}, \qquad \td\p = \td\p\big|_{s=0}.
\ee
This means that the dilaton $\p$ is given by
\ba\label{eq:retphi}
\p &= \td\p+ 2\l e^{-2\w} \pa_u \y \pa_v \y \\
&= \p\big|_{s=0} -4\pi \l s \q(-u) \q(v) \pa_u \y \pa_v \y.
\ea

As in \secref{sec:ein}, the advanced solution is  analogous to the retarded solution with the understanding that both the conformal factor and the change in $\phi$ are now supported in the {\it past} wedge.  
The desired linearized solution can then be obtained by subtracting the two solutions. In terms of the fields $\y,\td\p$, it precisely takes the form shown in \figref{fig:flowed}. From the flowed solution, it is straightforward to read off the Peierls bracket of $\s$ with any gauge-invariant observable.

In particular, we can read off the transformation of the initial data on a Cauchy slice passing through $\g$. As in \figref{fig:lcone}, we choose $\S$ at $t=0^-$ and $\S_s$ at $t=0^+$ and compare them in the retarded solution. It is clear from \Eqref{eq:retnew} that there is no change to $\y,\dot\y$. Also, since the geometry is identical to that of \figref{fig:lcone}, the change in the extrinsic curvature is again given by the BCP kink transformation \Eqref{eq:extchange}. Finally, for the dilaton, from \Eqref{eq:retphi} we find the following changes to the initial data:
\ba
\d\p &= -4\pi \l s \q(-x+0^+) \q(x+0^+) \pa_u \y \pa_v \y \big|_{t=0^+} = 0,\\
\d\dot\p &= -4\pi \l s \(\d(-x+0^+) \q(x+0^+) + \q(-x+0^+) \d(x+0^+)\) \pa_u \y \pa_v \y \big|_{t=0^+}\\
&= -2\pi \l s \d(x) \(\y'^2- \dot\y^2\) \Big|_{t=0^+},\label{eq:phi_ch}
\ea
where we remind the reader that $x$ measures proper distance along the Cauchy slice. The result \Eqref{eq:phi_ch} shows that the transformation generated by the geometric entropy in this higher-derivative theory is not just the BCP kink-transformation.  Indeed, the fact that $\phi$ and $\tilde \phi$ differ immediately implies that the BCP kink-transformation by itself fails to preserve the constraint equation \Eqref{uueom2} (and similarly for the $vv$ constraint).

To summarize, the geometric entropy flow in this theory is a generalization of the BCP kink-transformation given in covariant form by
\begin{align}
    \d K_{\perp\perp} &= 2\pi s \,\delta_{\Sigma}\(\gamma\), \\
    \d \dot\p &= -2\pi \l s \(\na_\mu \y \na^\m \y\)\delta_{\Sigma}\(\gamma\), \\
    \d \y &= \d \dot \y = \d h=\d\p=0,
\end{align}
where the notation used is that of \Eqref{eq:kink}.

\subsection{Theory 2}
\label{sub:t2}

The second model we consider has an additional scalar field and a more complicated higher-derivative interaction. Its action is given by
\begin{equation}\label{eq:action2}
	I_0 = \frac{1}{2} \int d^2 x \sqrt{-g}\[\phi R-(\na \y)^2 -(\na \r)^2 +\l \r \na_\m \na_\n \y \na^\m \na^\n \y \],
\end{equation}
where $\l$ is again treated as a perturbatively small coupling. In this model, we will only work to first order in $\l$.

As before, we can first analyze the general solution to the equations of motion. The equation $R=0$ gives us Minkowski space which we again chart using lightcone coordinates $(u,v)$. Each field is then decomposed in a power series in $\l$, e.g., $\y=\y_0+\l \y_1+\dots$. In these coordinates, the leading order equations of motion are
\begin{align}\label{eq:1storder}
	&\pa_u \pa_v \psi_0=0,\qquad 
	\pa_u \pa_v \r_0=0, \qquad \pa_u \pa_v \p_0 =0,\\
	&\pa_u^2 \p_0 + \(\pa_u \psi_0\)^2+\(\pa_u \r_0\)^2=0.\\
	&\pa_v^2 \p_0 + \(\pa_v \psi_0\)^2+\(\pa_v \r_0\)^2=0.
\end{align}
The above equations imply that each of the fields $\r_0,\psi_0,\p_0$ can be decomposed into a sum of left and right movers, e.g.,
\begin{equation}
    \r_0(u,v) = \r_{0R}(u)+\r_{0L}(v).
\end{equation}
We can think of $\r_0,\psi_0$ as free data, whereas $\p_0$ is fixed by solving the constraint equations. 

At first order in $\l$, the equations of motion can be written in a convenient form
\begin{align}
	&\pa_u\pa_v \td \r_1=0,\qquad \pa_u\pa_v \td \y_1=0,\qquad \pa_u\pa_v \td \p_1=0,\\
	&\pa_u^2 \td \p_1 + 2\(\pa_u \td \psi_0 \pa_u \td \psi_1+\pa_u \td \r_0 \pa_u \td \r_1\)=0\label{eq:phi1uu},\\
	&\pa_v^2 \td \p_1 + 2\(\pa_v \td \psi_0 \pa_v \td \psi_1+\pa_v \td \r_0 \pa_v \td \r_1\)=0,
\end{align}
where we have defined the fields
\begin{align}
    \td \r &= \r + \frac{\l}{4}\na_\m\y \na^\m \y = \r +\l \pa_u \y \pa_v\y,\\
    \td \y &= \y + \frac{\l}{2}\na_\m\y \na^\m \r = \y +\l \(\pa_u \y \pa_v\r+\pa_u \r \pa_v\y\),\\
    \td \p &= \p - \frac{\l}{2}\r\na_\m\y \na^\m \y = \p -2 \l \r \pa_u \y \pa_v\y,
\end{align}
in terms of which the equations take the simple form of JT gravity with minimally coupled matter fields to $O(\l)$. By the equations of motion, all of these fields take the form of a sum over left and right movers, e.g.,
\begin{align}\label{eq:w1}
	&\tilde{\r}_1=\tilde{\r}_{1R}(u)+\tilde{\r}_{1L}(v).
\end{align}

Having discussed the general solutions to the EOMs at $O(\l)$, we now move on to discussing the geometric entropy flow. The geometric entropy at $O(\l)$ is given by \cite{Dong:2017xht}
\begin{equation}
    \sg = 2\pi \phi - \pi \l \r \(\nabla \y\)^2 = 2\pi \tilde{\p}
\end{equation}
evaluated at its extremal point. The first step is to work out the EOMs for the modified action $I_s=I_0-s \sg$. From varying the dilaton, we have
\begin{equation}
    R= 4\pi s \d_\g,
\end{equation}
where we have again used extremality of $\g$. The geometry is as before and is given in conformal gauge by
\be
ds^2 = e^{2\w} du dv,
\ee
with $\w = \pi s\, \q\(-u\)\q\(v\)$ for the retarded solution.

At $O(\l^0)$, there are no source terms in any of the remaining EOMs and we end up with the same solution as in \figref{fig:lcone} for $\r,\y,\p$.
At $O(\l)$, the EOM for $\y$ is
\be
\na_\m \na^\m\y + \l \na_\m \na_\n (\r\na^\m \na^\n \y) = 2 \pi\l s \na_\m \(\d_\g \r \na^\m \y\).
\ee
Using the zeroth order solution and the geometry of the cone, we find that the above equation reduces to
\be
\pa_u \pa_v \td\y = 0.
\ee
Thus, the solution for $\td \y$ is given by \figref{fig:lcone}.

The EOM for $\r$ is
\be
\na_\m \na^\m\r + \fr{\l}{2} \na_\m \na_\n\y \na^\m \na^\n \y = -\pi\l s \d_\g \na_\m\y \na^\m \y.
\ee
Again, we find that this reduces to
\be
\pa_u \pa_v \td\r =0.
\ee
Thus, the solution for $\td \r$ is also given by \figref{fig:lcone}.

Finally, the EOMs obtained from varying the metric $g$ are
\bm
g_{\m\n} \[-\fr{1}{2} \na^2\p -\fr{1}{4}(\na\y)^2 -\fr{1}{4}(\na\r)^2 +\fr{\l}{4} \r \na_\a\na_\b\y \na^\a\na^\b\y\]
+\fr{1}{2} \na_\m\na_\n\p +\fr{1}{2} \na_\m\y \na_\n\y +\fr{1}{2} \na_\m\r \na_\n\r \\
+\fr{\l}{2} \Big\{\big[ \na_\m\y \na^\a (\r \na_\a \na_\n\y) +(\m\lra\n)\big] -\na^\a (\r \na_\a\y \na_\m \na_\n\y) \Big\} = \pi \l s \d_\g \r \na_\m\y \na_\n\y.
\em
The $(uv)$-component simplifies to
\be
\pa_u \pa_v \td\p = 0.
\ee
Thus, the solution for $\td \p$ is also given by \figref{fig:lcone}.

We can also verify that the $(uu)$-component, which is a constraint equation, is also satisfied. It reduces to
\be
\fr{1}{2} \pa_u \pa_u\td\p - \pa_u\w \pa_u\td\p +\fr{1}{2} (\pa_u\td\y)^2 +\fr{1}{2} (\pa_u\td\r)^2 =0.
\ee
This equation is identical to the unsourced EOM, except for the potential $\d$ function contained in $\pa_u \w$. However, the extremality condition (combined with $\td\p$ being a sum over left and right movers) ensures that this contribution vanishes and thus, the $(uu)$-component is satisfied.

Therefore, to summarize, the retarded solution is given by
\ba
\w &= \pi s\q(-u) \q(v),\qquad \td\y = \td\y\big|_{s=0},\qquad \td\r = \td\r\big|_{s=0},\qquad \td\p = \td\p\big|_{s=0}.
\ea
Expressed in terms of the original fields, the solution is
\ba\label{eq:psiret2}
\y &= \y\big|_{s=0} +2\pi \l s \q(-u) \q(v) \[\pa_u \r \pa_v \y + (u\lra v)\], \\
\r &= \r\big|_{s=0} +2\pi \l s \q(-u) \q(v) \pa_u \y \pa_v \y,\label{eq:rhoret2}\\
\p &= \p\big|_{s=0} -4\pi \l s \q(-u) \q(v) \r \pa_u \y \pa_v \y.\label{eq:phiret2}
\ea

Again, the advanced solution is identical to the retarded solution except that $\w=\pi s$ in the past wedge only. The flowed solution then takes the form of \figref{fig:flowed} in terms of the fields $\td \r,\td \y,\td\p$. From the flowed solution, it is straightforward to read off the Peierls bracket of $\s$ with any gauge-invariant observable.

In particular, we can read off the transformation of the initial data on a Cauchy slice passing through $\g$. As in \figref{fig:lcone}, we choose $\S$ at $t=0^-$ and $\S_s$ at $t=0^+$ and compare them in the retarded solution. Since the geometry is identical to that of \figref{fig:lcone}, the change in the extrinsic curvature is again given by the BCP kink transformation, \Eqref{eq:extchange}. For the other fields, we can use \Eqref{eq:psiret2}, \Eqref{eq:rhoret2} and \Eqref{eq:phiret2}.
In summary, we have
\ba
\d h&= \d\r = \d\y = \d\p = 0,\\
\d K_{\perp\perp}&= 2\pi s \,\delta_{\Sigma}\(\gamma\),\\
\d\dot\r &= \pi \l s \d(x) \(\y'^2- \dot\y^2\) \Big|_{t=0^+}=\pi \l s \na_\m \y\na^\m\y\, \d_\S(\g),\\
\d\dot\y &= 2\pi \l s \d(x) \(\r'\y'- \dot\r \dot \y\) \Big|_{t=0^+}=2\pi \l s \na_\m \r\na^\m\y\, \d_\S(\g),\\
\d\dot\p &= -2\pi \l s \d(x) \r \(\y'^2- \dot\y^2\) \Big|_{t=0^+}=-2\pi \l s \r \na_\m \y\na^\m\y\, \d_\S(\g),
\ea
where we remind the reader that $x$ measures proper distance along the Cauchy slice. The above transformation is the geometric entropy flow in this higher-derivative theory, generalizing the BCP kink-transformation. Again, one can check that the BCP kink-transformation on its own would not satisfy the constraint equations.

\section{Discussion}

\label{sec:discussion}

We now comment on several aspects of our results above as well as possible future extensions.

\subsection*{Modular flow}

Here, we note a connection of our results to modular flow in AdS/CFT which will be further elaborated on in Ref.~\cite{xi}. Inspired by Ref.~\cite{Jafferis:2015del}, it was suggested in Ref.~\cite{Bousso:2020yxi} that in an appropriate bulk semiclassical limit the geometric entropy flow is related by holographic duality to modular flow in the dual CFT.\footnote{Closely related comments also appear in Refs.~\cite{Jafferis:2014lza,Jafferis:2015del,Lewkowycz:2018sgn, Chen:2018rgz,Faulkner:2018faa,Bousso:2019dxk}.
In particular, given a boundary state $\ket{\psi}$ which induces a density matrix $\rho_R$ on region $R$, one expects that the geometric entropy flow approximately corresponds to the one-sided modular flow given by
\begin{equation}
	\ket{\psi(s)} = e^{-i K_R s} \ket{\psi},
\end{equation}
where $K_R = -\log \rho_R$. 
In addition, when the boundary conditions admit an appropriate symmetry, a regularized version of modular flow known as the Connes-cocycle (CC) flow was conjectured in Ref.~\cite{Bousso:2020yxi} to generate a kink at the HRT surface as in \Eqref{eq:kink} while also having an additional effect on asymptotic boundary conditions. 
Our result confirms these expectations and generalizes them to higher derivative theories.}

The BCP kink-transformation was in fact originally proposed as a holographic dual to Connes cocycle flow when the vacuum modular hamiltonian on the boundary takes a local form \cite{Bousso:2020yxi}. In Ref.~\cite{xi}, we will demonstrate that the BCP kink-transformation  approximates modular flow more generally when acting on an appropriate class of semiclassical bulk states. The proposal of Ref.~\cite{Bousso:2020yxi} is then a special case of this result.

 Let us also note that our calculations required only extremality and not minimality.
We may thus consider a situation with multiple extremal surfaces $\gamma_1$ and $\gamma_2$, with $\mathcal{A}(\gamma_1)<\mathcal{A}(\gamma_2)$. The true HRT surface in this context is $\gamma_1$, and flowing by the HRT area is dual (at leading order in $G$) to the boundary modular flow when acting on the same state.

However, we could also consider the BCP kink-transformation applied at $\gamma_2$, a non-minimal extremal surface. In the classical phase space, it is clear that this state represents a flow by the non-minimal area operator. We can give this a boundary interpretation by considering a coarse grained state of the form described in Ref.~\cite{Engelhardt:2021mue}. The coarse grained state $\tilde{\rho}_R$ can be prepared by acting with unitaries in the region between $\gamma_1$ and $\gamma_2$ so that the boundary subregion $R$ loses the information of this bulk subregion. It is then clear that the modular Hamiltonian of the coarse-grained state is given by the area operator of $\gamma_2$.

\subsection*{Singularities in the flowed solution}

The BCP kink-transformation in Einstein gravity leads to delta functions in the Weyl tensor along the lightcones of the HRT surface \cite{Bousso:2020yxi}. Here, we have seen that the generalization of the BCP kink-transformation in higher-derivative theories generally contains additional singularities. These singularities raise the potential question of whether the equations of motion are well-defined on the flowed solution.

At linearized order in $s$, if we consider the flow starting from a smooth state, then it is clear that all singularities come with a coefficient of $s$. As a result, one need not worry about, e.g., products of $\d$-functions since they can at most arise at $O(s^2)$. However,  products of $\d$-functions might potentially arise when integrating the flow to finite $s$.  It turns out that this does not occur in computing the equations of motion for the theories studied in \secref{sec:ein} and \secref{sec:hdev}, for which the action evaluated on the flowed solution also remains finite.  But it would be interesting to understand whether this remains true more generally.

In any case, higher derivative couplings should generally be treated as part of a low-energy effective theory where the couplings come with a length scale that determines the ultraviolet cutoff. In such a framework, it is appropriate to first smear all fields over length scales set by $\l$ before inserting them into the equations of motion or action.   Doing so will then yield smooth and finite results.   

\subsection*{Lorentzian derivation of geometric entropy} 
\label{sub:lor_vs_euc}

The geometric entropy $\sg$ is derived as the response of the Euclidean gravitational action to the introduction of a conical singularity}. Thus, it would be of interest to have an alternative derivation of $\sg$ as  the response of the Lorentzian action to the introduction of a Lorentzian conical singularity.  In particular, recall that the qualitative forms of the Euclidean and Lorentzian solutions are quite different. 
The Euclidean solution involves a triple expansion near the conical defect, whereas the Lorentzian solution involves a regular double expansion with additional non-analyticities.
This simplifies the Lorentzian calculation since, unlike the Euclidean calculation, one only needs to keep track of  terms in the variation that are formally of first order.

However,  finding such a Lorentzian derivation appears to be a difficult task for general higher derivative theories since the generalization of the BCP kink-transformation fails to take a simple, universal form in all theories as we discussed before. If one had a method to determine the transformation a priori (without using specific expressions for $\sg$), then it would be  straightforward to derive the generator. Moreover, one can do so in a fashion resembling the Euclidean derivation as demonstrated in Appendix~\ref{app:lorentzian}.

However, we also show in Appendix~\ref{app:extra} that, in our example theories, one can  identify other generators that seem to have similar types of flows, and from among which it does not appear easy to pick out the geometric entropy flow without knowledge of $\sigma$.


\subsection*{Importance of extremality} 
\label{sub:minimality}

Geometric entropies are computed by evaluating certain quantities at their extrema.
The Peierls bracket method is amenable to computing the flow of observables on phase space evaluated at their non-extremal location as well. However, in  our examples, the restriction imposed by the extremality constraint simplified the calculation by removing some source terms A similar simplification was  used in Ref.~\cite{Kaplan:2022orm} to derive the BCP kink-transformation. 

The flow generated by the area of a surface $\gamma_R$ has been much studied when $\gamma_R$ lies on the boundary of the gravitational system \cite{Donnelly:2016auv,Speranza:2017gxd,Chandrasekaran:2019ewn}
In that setting, the area was shown to generate a local boost of the gravitational system about $\gamma_R$.
One might describe such settings as `one-sided,' in the sense that the system lies entirely to one side of $\gamma_R$.  
In contrast, our analysis (and that of Ref.~\cite{Kaplan:2023oaj}) follows Refs.~\cite{Bousso:2020yxi,Kaplan:2022orm} in studying surfaces $\gamma_R$ that do not lie on boundaries, so that the gravitational system is `two-sided,' extending to both sides of $\gamma_R$.  
Nevertheless, as in Refs.~\cite{Bousso:2020yxi,Kaplan:2022orm} we find that similar boosts play an important role due to properties of the HRT surface. 
In particular, extremality of the geometric entropy functional at $\gamma_R$ ensures that the results for the one-sided phase space  continue to hold in the two-sided context.

It would thus also be interesting to analyze the flows generated by functionals evaluated at non-extremal surfaces.
A particular candidate  is the area functional evaluated at the quantum extremal surface of an evaporating black hole.

\subsection*{Implications for obtaining $S_{gen}$ from type II von Neumann algebras}

The argument of Refs.~\cite{Chandrasekaran:2022eqq,Kudler-Flam:2023qfl} for the type II nature of the algebra of local fields outside a black hole can be made directly from the gravitational constraints, whose association with bulk diffeomorphisms remains the same in general higher derivative theories.  But,  in the context of black holes with Killing horizons and asymptotic boundaries, its computation showing that the associated type II von Neumann entropy agrees with the generalized entropy $S_{gen}$ requires relating the (linearized) ADM Hamiltonian to the bulk one-sided Killing charge (the integral of $T_{ab}\xi^a n^b$ together with appropriate graviton contributions) and the geometric entropy $\sigma$.  This relation turns out to follow immediately from the fact that the bulk Killing charge generates one-sided asymptotic time translations up to a BCP kink transformation.  In particular, as argued in \secref{sec:intro}, in this context the flow generated by the geometric $\sigma$ then cancels the BCP kink-transformation, leaving only the desired asymptotic time-translation (perhaps also with a Yang-Mills charge-rotation as described in footnote~\ref{foot:charge}).  Our explicit results in \secref{sec:hdev} are consistent with this in that the departures of the geometric entropy flow identified in \secref{sec:hdev} from the BCP kink-transformation vanish identically in the presence of the above Killing symmetry.

However, as outlined in Section 4 of Ref.~\cite{Chen:2024rpx}, there should again be a type II algebra of local fields even in the absence of such a Killing symmetry.\footnote{In fact, the special case of asymptotic Killing horizons in the far past has already been understood in Ref.~\cite{Kudler-Flam:2024psh}.}   In contexts with an asymptotically AdS boundary, one expects any connection with $S_{gen}$ to proceed via the higher derivative generalization $K_{bndy}=\sigma + K_{\bulk}$ of the JLMS relation of Ref.~\cite{Jafferis:2015del}. 
In contexts where the contribution of $K_{\bulk}$ to boundary modular flow is small,\footnote{See discussion in Ref.~\cite{xi}.} our results thus predict that boundary modular flow will induce non-geometric singularities of the form found in our study of geometric entropy flows in \secref{sec:hdev}, in addition to the Weyl tensor shocks found in Einstein gravity \cite{Bousso:2020yxi}. 




\acknowledgments

We would like to thank Venkatesa Chandrasekaran, Tom Faulkner, Molly Kaplan, Jonah Kudler-Flam, Geoff Penington and Arvin Shahbazi-Moghaddam for useful discussions. 
We would also like to thank Chih-Hung Wu and Jiuci Xu for comments on a draft of this paper.
PR is supported in part by a grant from the Simons Foundation, by funds from UCSB, the Berkeley Center for Theoretical Physics; by the Department of Energy, Office of Science, Office of High Energy Physics under QuantISED Award DE-SC0019380, under contract DE-AC02-05CH11231 and by the National Science Foundation under Award Number 2112880. 
This material is based upon work supported by the Air Force Office of Scientific Research under award number FA9550-19-1-0360.

\appendix

\section{Covariant Phase Space Formalism}\label{app:covphase}

Here we provide a brief review of the covariant phase space formalism. For more details, we refer the reader to the recent review in Ref.~\cite{Harlow:2019yfa}.   This review largely serves to establish notation that we will use in later appendices. 

We begin by considering a local action
\be
I = \int_M L +\int_{\pa M} \ell
\ee
defined on a space of field configurations $\cC$.
Under a general variation $\dc\p^a$ of the dynamical  fields $\p^a(x)$, we have
\be
\dc I = \int_M \dc L +\int_{\pa M} \dc\ell.
\ee
We will also use $\dc$ to denote the exterior derivative on the configuration space $\cC$.\footnote{Note that this was denoted $\d$ in Ref.~\cite{Harlow:2019yfa}.} Note that
\be
\dc L = E_a \dc\p^a + d\q,
\ee
where $\q$ is the pre-symplectic potential current\footnote{This was denoted $\Q$ in Ref.~\cite{Harlow:2019yfa}.}; it is a 1-form on $\cC$ and a local $(d-1)$-form on spacetime, defined up to $dY$ for some local $(d-2)$-form $Y$. Thus
\be
\dc I = \int_M E_a \dc\p^a +\int_{\pa M} (\q + \dc\ell).
\ee

We now write
\be
\pa M = \G \cup \S_- \cup \S_+,
\ee
partitioning the full boundary $\pa M$ into 
the union of a spatial boundary $\G$ and the past/future boundaries $\S_-$, $\S_+$. Note that they meet at $\pa \G = \pa\S_- \cup \pa\S_+$. The variational principle requires
\be
(\q + \dc\ell)|_\G = dC
\ee
for some local $(d-2)$-form $C$ on $\G$. We now extend $C$ to all of $M$, so the RHS above should be replaced by $dC|_\G$. Note that $C$ is defined up to $dX$ for some local $(d-3)$-form $X$. Moreover, if we shift $\q$ by $dY$, we shift $C$ by $Y$. Thus
\ba
\dc I &= \int_M E_a \dc\p^a +\int_{\S_+ - \S_-} (\q + \dc\ell) + \int_{\G} dC\\
&= \int_M E_a \dc\p^a +\int_{\S_+ - \S_-} (\q + \dc\ell -dC).
\ea

Finally, define the pre-symplectic current
\be
\w := \dc (\q -dC) \, \big|_\tP = \dtP \[(\q -dC) |_\tP\],
\ee
where $|_\tP$ denotes the pullback to the pre-phase space $\tP$ defined by the equations of motion (i.e., restricting the base manifold to the submanifold $\tP$ of $\cC$ on which the equations of motion hold, and restricting the action of differential forms to tangent vectors on $\tP$), $\dtP$ is the exterior derivative on $\tP$, and we used that pullbacks commute with exterior derivatives. Note that $\w$ vanishes on $\G$:
\be
\w|_\G = \dc (\q + \dc\ell -dC) \, \big|_{\tP,\G} = 0.
\ee
 It is closed as a 2-form on $\tP$:
\be
\dtP \w = \dtP^2 \[(\q -dC) |_\tP\] = 0,
\ee
and it is also closed as a $(d-1)$-form on spacetime:
\be
d \w = \dc d\q \, \big|_\tP = \dc (\dc L - E_a \dc\p^a) \, \big|_\tP = -\dc E_a \wg \dc\p^a \, \big|_\tP = 0,
\ee
since $E_a = 0$ on $\tP$. We define the pre-symplectic form as
\be
\tW := \int_\S \w = \dc \lt.\(\int_\S \q -dC \) \,\rt|_\tP = \dc \lt.\(\int_\S \q - \int_{\pa\S} C\) \,\rt|_\tP,
\ee
where $\S$ is any Cauchy slice of $M$. Clearly, $\tW$ does not depend on the choice of $\S$, because $d \w=0$ and $\w|_\G = 0$. Moreover, $\tW$ is closed as a 2-form on $\tP$: $\dtP\tW =0$.

In general, $\tW$ may be written as $\dtP \tQ$; such a $\tQ$ is called a pre-symplectic potential. In particular, our $\tW$ can be written in terms of a pre-symplectic potential
\be
\tQ_\S := \lt.\(\int_\S \q -dC \) \rt|_\tP,
\ee
which may depend on $\S$.

The phase space $\cP$ is then defined as
\be
\cP := \tP / \wtd G,
\ee
where $\wtd G$ is the group generated by zero modes of $\wtd\W$. This quotient naturally defines a symplectic form $\W$ from $\wtd\W$ and a symplectic potential $\Q$ from a $\tQ$, satisfying $\W = \dP \Q$ where $\dP$ is the exterior derivative on $\cP$.

We can view $\W$ as a linear map sending vectors to 1-forms:
\be
\W(Y)(X) = \W(X,Y).
\ee
Since $\W$ is non-degenerate, this map has an inverse $\W^{-1}$ sending a 1-form $\s$ to a vector $\W^{-1}(\s)$. Thus we can also view $\W^{-1}$ as an anti-symmetric 2-vector that acts on a pair of 1-forms as $\W^{-1}(\s,\s') = \s(\W^{-1}(\s'))$.

An observable $g$ is a (real-valued) function on $\cP$, or equivalently a $\wtd G$-invariant function on $\tP$. It defines a vector field $X_g$ on $\cP$:
\be
X_g := \W^{-1}(\dP g),
\ee
whose inverse reads
\be\label{eq:gen}
\dP g = \W(X_g) = - X_g \cdot \W,
\ee
where $\cdot$ denotes the interior product inserting a vector into the first argument of a differential form ($X_g \cdot \W$ is sometimes written $X_g \lrcorner \W$ or $\iota_{X_g} \W$). Moreover, $X_g$ is known as a symplectic (or Hamiltonian) vector field.  It generates a symplectomorphism since
\be
\cL_{X_g} \W = X_g \cdot \dP \W + \dP (X_g \cdot \W) = \dP (-\dP g) =0,
\ee
where the first equality comes from Cartan's magic formula.

The change in another observable $f$ under an infinitesimal flow along $X_g$ is then given by the Poisson bracket defined as:
\be
\{f,g\} := \cL_{X_g} f = \dP f(X_g) = \W^{-1} (\dP f,\dP g) = \W(X_g, X_f).
\ee

\section{Poisson/Dirac Bracket Calculation}
\label{app:poisson}

One method to compute Poisson brackets in theories with perturbative higher-derivative terms  is to simply invert the symplectic form perturabtively. An alternative method which we illustrate here relies on the Dirac bracket.  We begin by adding extra degrees of freedom as would be appropriate for a higher-derivative theory with finite couplings.  We then work perturbatively and reduce the phase space by imposing constraints. 

As a warmup, we first demonstrate the method in a toy example of a particle with higher derivative term in the Lagrangian. We then follow the same method in our example theories to compute Dirac brackets with $\sigma$.

\subsection{Particle example}
Consider
\be
L = \fr{1}{2} \(\dot x^2 -\w^2 x^2 +\l \ddot x^2\).
\ee
We use the following procedure in the canonical formalism \cite{Ostrogradsky:1850fid,Henneaux:1994lbw,Cheng:2001du,andrzejewski2007canonical}. We first consider a 4-dimensional phase space with two coordinates $x,y$ where
\be\la{ydef}
y := \dot x,
\ee
and the corresponding momenta are
\ba
p_x &= \fr{\d L}{\d \dot x} := \fr{\pa L}{\pa \dot x} - \fr{d}{dt} \fr{\pa L}{\pa \ddot x} = \dot x - \l \dddot x,\\
p_y &= \fr{\pa L}{\pa \ddot x} = \l \ddot x.
\ea
At this point, we can insert the zeroth-order EOM $\ddot x=-\w^2 x$ into the first-order terms, obtaining
\ba
p_x &= \(1+ \l \w^2\) \dot x,\\
p_y &= -\l \w^2 x.
\ea
Here and below we work to $\cO(\l)$.
The first equation allows us to solve for $\dot x$ in terms of $p_x$:
\be
\dot x = \(1-\l \w^2\) p_x.
\ee
The second equation is a primary constraint on the phase space:
\be
c_2 := p_y +\l \w^2 x \approx 0,
\ee
where we follow Dirac in using $\approx$ to mean ``weakly equal'' (i.e., equal on the constraint surface). The definition \er{ydef} of $y$ is interpreted as another constraint:
\be
c_1 := y- \(1-\l \w^2\) p_x \approx 0.
\ee
These constraints are second-class since their (unconstrained) Poisson bracket is (weakly) nonzero:
\be
\{c_1,c_2\} = 1+\l \w^2,
\ee
where the (unconstrained) Poisson bracket is defined on the unconstrained 4-dimensional phase space as $\{x,p_x\}=\{y,p_y\}=1$. We then define the Dirac bracket
\be
\{f,g\}^* := \{f,g\} - \sum_{ij}\{f,c_i\} \(M^{-1}\)_{ij} \{c_j,g\} = \{f,g\} + \sum_{ij}(1-\l \w^2) \{f,c_i\} \e_{ij} \{c_j,g\},
\ee
where $M_{ij}:= \{c_i,c_j\}$ and $\e_{12}=-\e_{21}=1$. For example,
\be
\{x,\dot x\}^*= 1-2\l \w^2.
\ee
We then use the Dirac bracket to define the Poisson structure on the constraint surface.

It is easy to reproduce the Dirac bracket using the symplectic form. The unconstrained symplectic form on the 4-dimensional phase space is
\be
\W = \dP p_x \wg \dP x + \dP p_y \wg \dP y.
\ee
Imposing the constraints, we find the constrained symplectic form
\be
\W^* = \(1+ \l \w^2\) \dP\dot x \wg \dP x -\l \w^2 \dP x \wg \dP \dot x = \(1+ 2\l \w^2\) \dP\dot x \wg \dP x.
\ee
The Dirac bracket is simply its inverse:
\be
\{f,g\}^* = \W^{-1} (\dP f,\dP g).
\ee
In particular,
\be
\{x,\dot x\}^* = 1-2\l \w^2,
\ee
fully agreeing with the previous result.

\subsection{Theory 1}
We will now determine the transformation of initial data under $\s$ for Theory 1 from \secref{sub:t1} by calculating Dirac brackets. We work with the action
\be
I = \frac{1}{2} \int d^2 x \sqrt{-g}\[\phi R-(\na \y)^2 +\l \na_\m \na_\n \y \na^\m \na^\n \y\]
\ee
in Gaussian normal coordinates 
\begin{equation}
\label{eq:GNC}    
ds^2=-dt^2+h dx^2.\end{equation}
The Lagrangian is
\be
L = \fr{1}{2} \int dx \(\sqrt{-g}\[\phi R-(\na \y)^2 +\l \na_\m \na_\n \y \na^\m \na^\n \y\] -\pa_t\[\fr{\p \dot h}{\sqrt{h}}\]\),
\ee
where we have added the last term (which is a total derivative in $t$) to remove $\ddot h$. We work to $\cO(\l)$ unless otherwise stated.

First, the unconstrained phase space includes coordinates $h,\p,\y$, and
\be\la{xdef}
\x := \dot \y,
\ee
as well as their momenta
\be
\pi_h = \fr{\d L}{\d \dot h},\qqu
\pi_\p = \fr{\d L}{\d \dot \p},\qqu
\pi_\y = \fr{\d L}{\d \dot \y} - \pa_t \fr{\d L}{\d \ddot \y},\qqu
\pi_\x = \fr{\d L}{\d \ddot \y}.
\ee
At $\cO(\l^0)$, the momenta are
\be
\pi_h = -\fr{\dot\p}{2\sqrt{h}} +\cO(\l),\qu
\pi_\p = -\fr{\dot h}{2\sqrt{h}} +\cO(\l),\qu
\pi_\y = \sqrt{h} \dot\y +\cO(\l),\qu
\pi_\x = \cO(\l).
\ee
The  zeroth-order EOMs for $\p$ and $\y$
\ba
\ddot h &= \fr{\dot h^2}{2h},\\
\ddot\y &= \fr{-h\dot h \dot\y +2h \y''- h'\y'}{2h^2}
\ea
can be used to replace fields with two or more $t$-derivatives in the momenta. At $O(\l)$, we find
\ba
\pi_h &= -\frac{\dot{\phi }}{2 \sqrt{h}}+\lambda\frac{ \psi ' \left(\dot{\psi } h'+4 h \dot{\psi }'\right)+h \dot{\psi } \left(\dot{h} \dot{\psi }-2 \psi ''\right)-2 \dot{h} \psi '^2}{4 h^{5/2}},\\
\pi_\p &= -\frac{\dot{h}}{2 \sqrt{h}},\\
\pi_\y &= \sqrt{h} \dot{\psi }+\lambda\frac{ h' \left(3 \dot{h} \psi '-2 h \dot{\psi }'\right)+2 h \left(-\dot{h} \psi ''-\dot{h}' \psi '+2 h \dot{\psi }''\right)}{4 h^{5/2}},\\
\pi_\x &= \l \fr{-h\dot h \dot\y +2h \y''- h'\y'}{2h^{3/2}}.
\ea
We solve the first three equations for the velocities
\ba
\dot h &= -2 \sqrt{h} \pi _{\phi },\\
\dot\p &= -2 \sqrt{h} \pi _h -\lambda\frac{ \pi _{\psi } h' \psi '+2 h \left(\pi _{\psi } \psi ''-2 \psi ' \pi _{\psi }'+\pi _{\phi } \left(\pi _{\psi }^2-2 \psi '^2\right)\right)}{2 h^{5/2}},\\
\dot\y &= \frac{\pi _{\psi }}{\sqrt{h}} +\lambda\frac{ -2 h^2 \left(\pi _{\psi }''+\pi _{\phi } \psi ''+\psi ' \pi _{\phi }'\right)-2 \pi _{\psi } h'^2+h \left(\pi _{\psi } h''+h' \left(3 \pi _{\psi }'+2 \pi _{\phi } \psi '\right)\right)}{2 h^{7/2}},
\ea
whereas the last of the 4 equations becomes
\be
\pi_\x = \l \fr{2h\pi_\p \pi_\y +2h \y''- h'\y'}{2h^{3/2}}.
\ee
This is a primary constraint:
\be
c_2 = \pi_\x - \l \fr{2h\pi_\p \pi_\y +2h \y''- h'\y'}{2h^{3/2}}.
\ee
The definition \er{xdef} of $\x$ is interpreted as another constraint:
\be
c_1 = \x - \dot\y = \x - \fr{\pi_\y}{\sqrt{h}} +\cO(\l).
\ee
The (unconstrained) Poisson bracket is defined as
\be
\{h(t,x),\pi_h(t,y)\} = \{\p(t,x),\pi_\p(t,y)\} = \{\y(t,x),\pi_\y(t,y)\} = \{\x(t,x),\pi_\x(t,y)\} = \d(x-y).
\ee
From this we find
\be
\{c_1(t,x), c_2(t,y)\} = \d(x-y) +\cO(\l).
\ee
So these are second-class constraints, and the Dirac bracket is defined as
\ba\label{eq:dirac1}
\{f,g\}^* &:= \{f,g\} - \int dz dw \sum_{ij}\{f,c_i(t,z)\} \(M^{-1}\)_{ij, zw} \{c_j(t,w),g\} \\
&= \{f,g\} + \int dz \sum_{ij} \{f,c_i(t,z)\} \e_{ij}  \{c_j(t,z),g\} (1+\cO(\l)).
\ea

Now using
\be
\s = 2\pi\p +\pi \l \(\dot\y^2 -\fr{\y'^2}{h}\) \bigg|_{x=t=0} = 2\pi\p +\pi \l \(\fr{\pi_\y^2 -\y'^2}{h}\) \bigg|_{x=t=0},
\ee
we find the (unconstrained) Poisson brackets:
\ba
\{h(0,x), \s\} &= 0,\\
\{\dot h(0,x), \s\} &= 4\pi \sqrt{h} \d(x),\\
\{\p(0,x), \s\} &= 0,\\
\{\dot\p(0,x), \s\} &= -2\pi \l \fr{\y'^2}{h^{3/2}} \d(x),\\
\{\y(0,x), \s\} &= 2\pi \l \fr{\pi_\y}{h} \d(x),\\
\{\dot\y(0,x), \s\} &= 0.
\ea

Using \Eqref{eq:dirac1}, the Dirac brackets with $\s$ are given by
\ba
\{h(0,x), \s\}^* &= 0,\\
\{\dot h(0,x), \s\}^* &= 4\pi \sqrt{h} \d(x),\\
\{\p(0,x), \s\}^* &= 0,\\
\{\dot\y(0,x), \s\}^* &= 0,\\
\{\dot\p(0,x), \s\}^* &=  2\pi \l \(\fr{h\dot\y^2 -\y'^2}{h^{3/2}}\) \d(x),\\
\{\y(0,x), \s\}^* &=0,
\ea
where we have used $\{c_2(0,z),\s\} =2\pi\l \d(z) \pi_\y /\sqrt{h}$.

Now all 6 Dirac brackets agree with our previous results. In particular, they agree with the initial data change $\d\y=0$ and also lead to
\be
\d\dot\p =-2\pi \l s \d(x) \(\fr{\y'^2- h\dot\y^2}{h^{3/2}} \)
\ee
which agrees with the previous result \eqref{eq:phi_ch} after setting $h=1$.

Having completed the calculation, it is important to discuss what we are doing conceptually. The fields for which we have obtained the transformation are not coordinates on the gauge-invariant phase space since there is still residual gauge freedom. For example, in JT gravity, the phase space is just two-dimensional \cite{Harlow:2018tqv}. However, if we work at the level of the pre-phase space, then the residual gauge freedoms are zero modes of the pre-symplectic form $\tW$ defined in Appendix~\ref{app:covphase} so that $\tW$ cannot be inverted to obtain Poisson brackets.

Since the zero-modes correspond to gauge-directions, the Poisson bracket of two gauge-invariant functions is nevertheless well-defined.  One could thus proceed by restricting attention to gauge-invariant functions on the pre-phase space and, indeed, this is effectively the strategy employed in our study of the Peierls bracket.

However, we have followed a different approach in this appendix.   Instead of using the original $\tW$ from Appendix~\ref{app:covphase}, we instead imposed  the Gaussian normal coordinate condition \Eqref{eq:GNC} in the {\it action} and then used the gauge-fixed action to define a new symplectic structure.  In particular,because the gauge-fixed action does not lead to the full equations of motion, the constraint equations ($E_{tt}$ and $E_{tx}$) have not been imposed at this point.  While such a gauge-fixing does not change Poisson brackets of gauge invariant quantities, it removes the zero-modes in $\tW$  associated with {\it all} gauge transformations, even the residual ones.  Indeed, while the residual gauge transformations still remain zero-modes of the gauge-fixed action, they now define {\it global} symmetries of that action.   While the Poisson bracket of two gauge-invariant functions is unchanged by this gauge-fixing, other Poisson brackets now depend on the choice of gauge.

As mentioned above, the fact that residual gauge symmetries do not define zero-modes of the new symplectic structure described in the previous paragraph is related to the fact that variations of the gauge-fixed action fail to yield the full set of equations of motion and, in particular, to the fact that the constraint equations ($E_{tt}$ and $E_{tx}$) no longer follow from the gauge-fixed action.  As a result, the gauge-fixed symplectic structure does not respect the constraints, and Poisson brackets with the constraints do not vanish. Instead, Poisson brackets with the constraints generate residual gauge transformations which preserve the Gaussian normal gauge condition \Eqref{eq:GNC}.  Such gauge-fixings are thus closely related to the well-known approach of Dirac \cite{Dirac}.   See e.g.\ \cite{Marolf:1993af} for a general discussion of this phenomenon.

\subsection{Theory 2}

We now repeat the above analysis for the second example theory  (from \secref{sub:t2}). From the action, we can read off the conjugate momenta in Gaussian normal gauge which are given by
\begin{align}
  \pi_h  =&-\frac{\dot{\phi }}{2 \sqrt{h}}+\frac{\lambda  \r  \left(\psi ' \left(\dot{\psi } h'+4 h \dot{\psi }'\right)+h \dot{\psi } \left(\dot{h} \dot{\psi }-2 \psi ''\right)-2 \dot{h} \left(\psi '\right)^2\right)}{4 h^{5/2}}\\
\pi_\phi=&-\frac{\dot{h}}{2 \sqrt{h}}\\
\begin{split}\pi_\y=&\sqrt{h} \dot{\psi }+\frac{\lambda  \left(2 h \left(\psi ' \left(\dot{\r } h'-2 \dot{h} \r '\right)+h \left(\dot{\r } \left(\dot{h} \dot{\psi }-2 \psi ''\right)+4 \dot{\psi }' \r '\right)\right)\right)}{4 h^{5/2}}\\
&+\frac{\l\r  \left(h' \left(3 \dot{h} \psi '-2 h \dot{\psi }'\right)+2 h \left(-\dot{h} \psi ''+\dot{h}' \left(-\psi '\right)+2 h \dot{\psi }''\right)\right)}{4 h^{5/2}}\end{split}\\
\pi_\r=&\sqrt{h} \dot{\r }
\end{align}
The canonical Poisson brackets e.g.\ $\{\phi(x),\pi_{\phi}(y)\}=\d\(x-y\)$
then follow immediately.
In order to compute the flow generated by $\s$ on the fields of interest, we need to invert the above expressions to solve for the zeroth order momenta, i.e., $\dot{\phi},\dot{h},\dot{\y},\dot{\r}$.
At $O(\l)$, this can be done simply by re-expressing the $O(\l)$ terms on the right-hand side perturbatively in terms of the zeroth order momenta and discarding higher order terms.

Now the constraints to be imposed are
\begin{align}
    c_1&=\xi-\dot{\y}=\xi-\frac{\pi _{\psi }}{\sqrt{h}} +O(\l)\\
c_2&=\pi_\xi-\frac{\lambda  \r  \left(2 h \left(\psi ''+\pi _{\psi } \pi _{\phi }\right)-h' \psi '\right)}{2 h^{3/2}}
\end{align}
The bracket between the constraints is
\begin{equation}
    \{c_1(x),c_2(y)\}= \d(x-y)+O(\l),
\end{equation}
so these are again second-class constraints, and we need to use the Dirac bracket.
Now using
\be
\s = 2\pi\p +\pi \l \(\dot\y^2 -\frac{\y'^2}{h}\) \bigg|_{x=t=0} = 2\pi\p +\pi \l \r \(\frac{\pi_\y^2 -\y'^2}{h}\) \bigg|_{x=t=0},
\ee
we find the following Dirac brackets with $\s$:
\begin{align}
\{h(0,x), \s\}^* &= 0,\\
\{\dot h(0,x), \s\}^* &= 4\pi \sqrt{h} \d(x),\\
\{\p(0,x), \s\}^* &= 0,\\
\{\dot\p(0,x), \s\}^* &= \frac{2 \pi  \lambda   \r \left(h \dot \psi^2-\psi'^2\right)}{h^{3/2}}\delta (x),\\
\{\y(0,x), \s\}^* &=0\\
\{\dot\y(0,x), \s\}^* &= \frac{2 \pi  \lambda   \left(\psi' \r'-h \dot\psi \dot \r\right)}{h^{3/2}}\delta (x),\\
\{\r(0,x), \s\}^* &=0\\
\{\dot\r(0,x), \s\}^* &= \frac{\pi  \lambda  \left(\psi'^2-h \dot \psi^2\right)}{h^{3/2}}\delta (x).
\end{align}
These results agree with those  found earlier using the Peierls bracket method.

\section{Susceptibility and Generators} 
\label{app:lorentzian}

Suppose an oracle has handed us the generalization of the BCP kink-transformation for general higher-derivative theories, and we wish to use this information to construct the geometric entropy. As described in Appendix~\ref{app:covphase}, given a linearized transformation defined by a vector field $X_g$ on phase space, we can obtain the generator straightforwardly by using \Eqref{eq:gen}. 

However, it is also useful to describe an alternate method here that resembles the Euclidean derivation or geometric entropy in Ref.~\cite{Dong:2019piw}. As reviewed in \secref{sub:geometric}, the geometric entropy is computed as a response of the Euclidean action to the insertion of a conical defect. A natural guess is to mimic the same in Lorentzian signature with the retarded/advanced solution playing the role of the spacetime deformed by a conical defect.
Such a response function will be called a susceptibility, and we will discuss its relation to the generator of a given transformation below.

For a general observable $g$, in the process of computing the Peierls bracket, we deform the action to
\be\la{deformed}
I_s = I_0 - s g
\ee
and consider its deformed EOMs. We will specialize to the case of current interest where $g$ is a compactly supported observable on some $\S_0$ (which can be called $t=0$), i.e., $g$ is a spatial integral of $\p^a$ over a compact region $R$ (without boundary) on $\S_0$. To linear order in $s$, any deformed solution $\p^a_s$ can be written as
\be
\p^a_s = \p^a_0 +s h^a,
\ee
where $\p^a_0$ is a solution to the original EOMs. For any $h^a$, the corresponding configuration-space vector is
\be
X^{h} := \int d^dx h^a(x) \fr{\d}{\d \p^a(x)}.
\ee
We will use $h_R$ ($h_A$) to denote the retarded (advanced) $h$. The Hamiltonian vector field $X_g$ of $g$ is then given by
\be
X_g = X^{h_R} - X^{h_A}.
\ee

Recall the definition of the pre-symplectic potential
\be
\tQ_\S := \lt.\(\int_\S \q -dC \) \rt|_\tP.
\ee
Since the flows generated by the individual vector fields $X^{h_R}, X^{h_A}$ do not preserve the pre-phase space $\tP$, it is useful to define
\be
\tQ_{\S}^{\cC} := \int_\S \(\q -dC\),
\ee
which is a 1-form on the configuration space $\cC$ and from which $\tQ_\S$ is obtained by pullback to $\tP$: $\tQ_\S= \tQ_{\S}^{\cC}\big|_\tP$.
For a given $h_R$, we may now define a susceptibility $\td g$ as
\be
\td g := \lt.\[X^{h_R} \cdot \(\tQ_{\S_{0^+}}^{\cC} - \tQ_{\S_{0^-}}^{\cC}\)\] \rt|_{\tP} = \lt.\( X^{h_R} \cdot \tQ_{\S_{0^+}}^{\cC} \)\rt|_{\tP}.
\ee
Now, the following condition is useful in relating the susceptibility to the generator.
\begin{nlemma}
\label{lem1}
If
\be
\label{eq:Lemma1C}
\lt.\(\cL_{X^{h_R}} \tQ_{\S_{0^+}}^{\cC}\) \rt|_{\tP} =0,
\ee
then $\td g=g$ up to an additive constant.
\end{nlemma}
\begin{proof}
$\td g=g$ up to an additive constant is equivalent to $\dtP \td g = \dtP g$, which we now prove:
\ba
0 &= \lt.\(\cL_{X^{h_R}} \tQ_{\S_{0^+}}^{\cC}\) \rt|_{\tP} = \lt.\(X^{h_R} \cdot \dc \tQ_{\S_{0^+}}^{\cC}\)  \rt|_{\tP} + \lt.\dc \(X^{h_R} \cdot \tQ_{\S_{0^+}}^{\cC}\) \,\rt|_{\tP}\\
&= \lt.\((X^{h_R}-X^{h_A}) \cdot \dc \tQ_{\S_{0^+}}^{\cC}\)  \rt|_{\tP} + \lt.\dc \((X^{h_R}-X^{h_A}) \cdot \tQ_{\S_{0^+}}^{\cC}\) \,\rt|_{\tP}\\
&= \lt.\(X_g \cdot \dc \tQ_{\S_{0^+}}^{\cC}\) \rt|_{\tP} + \lt.\dc \(X_g \cdot \tQ_{\S_{0^+}}^{\cC}\) \,\rt|_{\tP}\\
&= X_g \cdot \dtP \tQ_{\S_{0^+}} + \dtP \(X_g \cdot \tQ_{\S_{0^+}}\)\\
&= X_g \cdot \tW + \dtP \td g = -\dtP g + \dtP \td g.\la{gtdminusg}
\ea
\end{proof}

Thus, if the above condition \er{eq:Lemma1C} is satisfied, the susceptibility  provides an alternate computation of the generator given a transformation. If, instead of vanishing, $\cL_{X^{h_R}} \tQ_{\S_{0^+}}^\cC$ turns out to be equal to an exact form, then we can add a total derivative to the action to satisfy the condition \er{eq:Lemma1C}. It must be a closed form since the transformation is a symplectomorphism, but there may be topological obstructions to it being an exact form.\footnote{In such situations, there can be symplectomorphisms with no generator. For instance, for a particle living on a circle, the transformation $p\to p+\a$ is a symplectomorphism with no generator. This is because the position $q$ is not a well-defined observable on phase space, only periodic functions of $q$ are.}

For future reference, we note that the above lemma in fact holds point by point in the pre-phase space $\tP$; i.e., if \Eqref{eq:Lemma1C} holds on any solution, then the Hamiltonian vector field of $\tilde g$ agrees with $X_g$ at that solution.

We can now discuss whether the condition of the above lemma is satisfied for the case of interest. While it is satisfied in simple theories like Einstein gravity and JT gravity, it generally fails in higher-derivative theories using the most obvious choice of pre-symplectic potential. We illustrate this in both our example theories. However, in both cases we can obtain a pre-symplectic potential that satisfies the desired condition by adding total derivative terms and, upon doing so, one recovers the correct geometric entropy.

\subsection{Theory 1}
From the action \Eqref{eq:action1}, we find
\bm\label{prepot}
    \wtd\Theta^\cC=\int \sqrt{-g} n_{\m} \lt\{\fr{1}{2}\[\phi\(g^{\m\a}\nabla^\b-g^{\a\b}\na^{\m}\)\Delta g_{\a\b}+\(\nabla^\m \phi g^{\a\b}-\nabla^\a\phi g^{\m\b}\)\Delta g_{\a\b}\]\right.\\
    \left.- \na^{\m} \y\Delta \y +\l\na^{\m}\na^{\n}\y \Delta\(\na_{\n}\y\) -\l\na_{\n}\na^{\n}\na^{\m}\y \Delta\y -\l \(\na^\m \na^\a\y \na^\b\y - \fr{1}{2} \na^\a \na^\b\y \na^\m\y\) \D g_{\a\b}\rt\},
\em
where $n^\m$ is past-pointing. Here we ignored a possible boundary term $\int_{\pa\S} C$ in the definition of $\wtd\Theta^\cC$, but such a boundary term would vanish when we act with $\cL_X$ or the interior product with $X$ later. In Gaussian normal coordinates, $n^\mu=(-1,0)$ and $\wtd\Theta^\cC$ becomes
\bm
    \wtd\Theta^\cC=-\int dx \sqrt{h} \[\(\frac{\dot{\phi}-K\phi}{2h}\)\Delta h -\phi\Delta K-\dot{\y}\Delta\y -\l \ddot{\y} \Delta \dot{\y}+\l\(\frac{\dot{\y}'- K\y'}{h}\)\Delta \y'\rt.\\
\lt. +\l\fr{2h \(2K\y'^2 -K h \dot\y^2 +\y'' \dot\y -2\y'\dot\y'\)-h'\y' \dot\y}{4h^3} \D h \].
\em

Now denote the transformation generated by $\s$ as a vector field $X$ on the pre-phase space. From the results of \secref{sub:t1}, the only non-trivial transformations on $\S_0$ (a time slice passing through the HRT surface) are given by
\begin{align}
    \cL_X K&= 2\pi \d(x)/\sqrt{h}\\
    \cL_X \dot\p &= -2\pi \l \fr{\y'^2- h\dot\y^2}{h^{3/2}} \d(x).
\end{align}
Thus, we find
\ba
\cL_X \wtd\Theta_{\S_0} &= \lt.2\pi \l \[\(\fr{\y'^2- h\dot\y^2}{2h^2} \)\dP h -\dot\y \dP\dot\y + \fr{\y'}{h} \dP \y' -\(\fr{2\y'^2-h \dot \y^2}{2h^2}\) \dP h\]\rt|_{x=0}\\
&= \lt.\pi \l \dP\(\fr{\y'^2}{h} -\dot\y^2\)\rt|_{x=0}.
\ea
As expected, we then obtain
\be
\cL_X \wtd\W = \cL_X \dP \wtd\Theta_{\S_0} = \dP \cL_X \wtd\Theta_{\S_0} = 0,
\ee
confirming that $X$ is a symplectomorphism. Moreover, since $\cL_X \wtd\Theta_{\S_0}$ is an exact form, we can add in a total derivative term to the action to satisfy $\cL_X \wtd\Theta_{\S_0}=0$.

Alternatively, we can find $g$ by using an identity derived in \er{gtdminusg}:
\be
\D(\td g-g) = \lt.\(\cL_{X^{h_R}} \tQ_{\S_{0^+}}^{\cC}\) \rt|_{\tP} = \cL_X \wtd\Theta_{\S_0} = \lt.\pi \l \dP\(\fr{\y'^2}{h} -\dot\y^2\)\rt|_{x=0}.
\ee
We have
\be
\td g= \lt.\( X^{h_R} \cdot \tQ_{\S_{0^+}}^{\cC} \)\rt|_{\tP} = X\cdot \wtd\Theta_{\S_0} = \lt.2\pi \phi\rt|_{x=0},
\ee
and thus,
\be
g= \lt.2\pi \phi - \pi \l \(\fr{\y'^2}{h} -\dot\y^2\)\rt|_{x=0} = \lt.2\pi \(\phi - \fr{\l}{2} \na_\m\y \na^\m\y\)\rt|_{x=0}
\ee
which is precisely the geometric entropy $\s$.


\subsection{Theory 2}
Working again in Gaussian normal coordinates and, moreover, taking $h=1$ and using the residual gauge freedom to simplify equations, we find
\begin{equation}
    \tQ^\cC=-\int \[-\phi\Delta K-\(\dot{\y} +\l\dot{\w}\ddot{\y}-\l\w'\(\dot{\y}'-K\y'\)\)\Delta\y-\dot{\w}\Delta\w -\l \w \ddot{\y} \Delta \dot{\y}+\l\w\(\dot{\y}'-K\y'\)\Delta \y'\].
\end{equation}
Computing the Lie derivative  yields
\begin{equation}
    \mathcal{L}_{X} \tQ = \dP\(X\cdot \tQ_{\S_0}\) +  X\cdot \dP \tQ_{\S_0}.
\end{equation}
Again using the results of \secref{sub:t2}, the only non-vanishing contribution to $X\cdot \tQ_{\S_0}$ comes from the first term in $\tQ$ and takes the form $2\pi \phi$. For the second term $X\cdot \dP \tQ_{\S_0}$, we find $-2\pi\dP\phi$ (which cancels the previous contribution) plus some additional contributions, leading to
\begin{equation}
    \mathcal{L}_X \tQ_{\S_0} = \dP\(\pi\l \w_0 \na_\m\y_0\na^\m\y_0\).
\end{equation}
Thus, the condition \er{eq:Lemma1C} is not satisfied, but the above results show that we can again deal with this by adding in a total derivative term to the action.  Once we do so, the resulting susceptibility agrees with the geometric entropy just as for Theory 1.

\section{More General Generators}
\label{app:extra}

This appendix studies one-parameter families of observables that generalize the geometric entropies of the example theories studied in \secref{sub:t1} and \secref{sub:t2}.  Our goal is to 
compare the flows generated by such observables with the geometric entropy flow of each theory. The comparison shows that it appears difficult to find a simple, universal prescription for distinguishing the geometric entropy flow with other flows.

\subsection{Theory 1}

Consider the one-parameter family of observables
\be
\c = 2\pi \phi - (1+\a) \pi\l \na_\m \psi \na^\m \psi \Big|_\g,
\ee
where $\g$ is defined to be an extremal surface with respect to $\c$. Note that $\a=0$ gives $\c=\s$. Adding $\c$ to the action to calculate the Peierls brackets, we find
\begin{equation}
I_s = \frac{1}{2} \int d^2 x \sqrt{-g}\[\phi R-(\na \y)^2 +\l \na_\m \na_\n \y \na^\m \na^\n \y\]  -s \c.
\end{equation}
To find the retarded solution, we solve the EOMs resulting from varying $I_s=I_0-s\c$ with respect to $g,\p,\y$.

Since the analysis is quite similar to the main text, we will  simply state the result. Up to $\cO(\l)$, the retarded solution in conformal gauge is
\ba
\w &= \pi s \q(-u) \q(v),\\
\y &= \y\big|_{s=0} + 2\a \pi \l s \[\d(u) \q(v) \pa_v\y(0) - \q(-u) \d(v) \pa_u\y(0)\], \\
\td\p &= \td\p\big|_{s=0} + 4\a\pi \l s \q(-u) \q(v) \pa_u\y(0) \pa_v\y(0).
\ea
This means that $\p$ takes the form
\ba
\p &= \td\p+ 2\l e^{-2\w} \pa_u \y \pa_v \y \\
&= \td\p\big|_{s=0} + 4\a\pi \l s \q(-u) \q(v) \pa_u\y(0) \pa_v\y(0) + 2\l e^{-2 \pi s \q(-u) \q(v)} \pa_u \y \pa_v \y\\
&= \p\big|_{s=0} -4\pi \l s \q(-u) \q(v) \pa_u \y \pa_v \y + 4\a\pi \l s \q(-u) \q(v) \pa_u\y(0) \pa_v\y(0).
\ea

In this theory, $\a=0$ is a somewhat special point in that the singularities in $\y$ vanish at $\a=0$.  It is thus of interest to see if this feature persists in our second example theory, to which we now turn.

\subsection{Theory 2}
Consider the action
\begin{equation}
I_s = \frac{1}{2} \int d^2 x \sqrt{-g}\[\phi R-(\na \y)^2 -(\na\r)^2 +\l \r \na_\m \na_\n \y \na^\m \na^\n \y\]  -s \c,
\end{equation}
and the one-parameter family of observables
\be
\c = 2\pi \phi - (1+\a) \pi\l \r \na_\m \psi \na^\m \psi \Big|_\g,
\ee
where $\g$ is defined to be an extremal surface with respect to $\c$. Again $\a=0$ gives $\c=\s$.

Repeating an analysis similar to that in the main text, we find that the retarded solution in conformal gauge is given by
\ba
\w &= \pi s\q(-u) \q(v),\\
\td\y &= \td\y\big|_{s=0} + 2\a\pi \l s \r(0) \[ \d(u) \q(v) \pa_v \y(0) - \q(-u) \d(v) \pa_u \y(0) \],\\
\td\r &= \td\r\big|_{s=0} + 2\a\pi \l s \q(-u) \q(v) \pa_u\y(0) \pa_v\y(0),\\
\td\p &= \td\p\big|_{s=0} + 4\a\pi \l s \q(-u) \q(v) \r(0) \pa_u\y(0) \pa_v\y(0).
\ea
In terms of the original fields, the solution takes the form
\ba
\y &= \y\big|_{s=0} +2\pi \l s \q(-u) \q(v) \[\pa_u \r \pa_v \y + (u\lra v)\] \\
&\qqu\qu + 2\a\pi \l s \r(0) \[ \d(u) \q(v) \pa_v \y(0) - \q(-u) \d(v) \pa_u \y(0) \],\\
\r &= \r\big|_{s=0} +2\pi \l s \q(-u) \q(v) \pa_u \y \pa_v \y + 2\a\pi \l s \q(-u) \q(v) \pa_u\y(0) \pa_v\y(0),\\
\p &= \p\big|_{s=0} -4\pi \l s \q(-u) \q(v) \r \pa_u \y \pa_v \y + 4\a\pi \l s \q(-u) \q(v) \r(0) \pa_u\y(0) \pa_v\y(0).
\ea

The question of whether there is generally a natural prescription for the retarded solution for $\c=\s$ can be phrased here as follows: how is $\a=0$ special? Based on the results above, it might be tempting to guess that,  to determine the retarded solution for $\s$ (corresponding to $\a=0$ above), we only need to require that none of $\y,\r,\p$ have a shockwave on the horizon. But it is easy to find another example where this is not sufficient: for general $\a$ we can simply consider another generator $\c=\s - \a \l \r |_\g$, in which case the retarded solution is the same as the $\a=0$ solution above, except that $\r$ contains an additional term $\fr{\a}{2} \l s \q(-u)\q(v)$. This solution has no shockwave either. A similar shockwave-free family of generators for Theory 1 is given by $\c=\s - \a \l \y |_\g$.
Therefore, it seems difficult to find a simple, universal prescription for the explicit form of the geometric entropy flow in general theories.

\section{Wald Entropy Generates BCP Kink Transform in Killing Horizon Cases}
\label{app:sigmaW}

In cases of Killing horizons, the geometric entropy reduces to the Wald entropy, and the associated flow is a BCP kink transform. We now show this.

We use the notation $L=\cL \ve$ where $\ve=\sqrt{-g} dx^0 \wg dx^1 \wg \cd \wg dx^{D-1}$ is the volume form.

\subsection{f(Riemann) gravity}

For simplicity, first consider f(Riemann) gravity. We view $\cL$ as a function of $g_{\m\n}$ and $R_{\m\r\n\s}$. Note that $g^{\m\n}$ is considered a function of $g_{\m\n}$. Thus
\be
\D \cL = \fr{\pa \cL}{\pa g_{\m\n}} \D g_{\m\n} + \fr{\pa \cL}{\pa R_{\m\r\n\s}} \D R_{\m\r\n\s}
= \fr{\pa \cL}{\pa g_{\m\n}} \D g_{\m\n} - 2 \fr{\pa \cL}{\pa R_{\m\r\n\s}} \D g_{\m\n;\r\s}.
\ee
We then find
\be
\D L = \ve \(\fr{\pa \cL}{\pa g_{\m\n}} +\fr{1}{2} \cL g^{\m\n}\) \D g_{\m\n} -2 \ve \fr{\pa \cL}{\pa R_{\m\r\n\s}} \D g_{\m\n;\r\s}.
\ee
Using $\D L = E_a \dc\p^a + d\q$, we find
\be\la{qfriem}
\q = \td\q \cdot \ve,\qqu
\td\q^\s = -2 \fr{\pa \cL}{\pa R_{\m\r\n\s}} \D g_{\m\n;\r} +2 \(\na_\r\fr{\pa \cL}{\pa R_{\m\s\n\r}}\) \D g_{\m\n}.
\ee

We now show
\be\la{lqcond}
\lt.\(\cL_{X^{h_R}} \tQ_{\S_{0^+}}^{\cC}\) \rt|_{\tP} =0,
\ee
where $h_R$ is the retarded solution obtaining by removing a piece of the Rindler wedge. We will use Gaussian normal coordinates $(t,x,y^i)$ near $\S_{0}$, such that $t$ is the proper time away from $\S_{0}$ and $x$ is the proper distance on $\S_{0}$ away from $\g$. Note that for the data on $\S_{0^+}$ (defined to be $t=0^+$), $X^{h_R}$ changes $K_{xx}$ by
\be\la{changekxx}
\cL_{X^{h_R}} g_{xx,t} = X^{h_R} \cdot \D g_{xx,t} = 2 X^{h_R} \cdot \D K_{xx} = 4\pi \d(x),
\ee
while keeping everything else (the induced metric $h_{ab}$, all other components of $K_{ab}$, and the Riemann tensor) unchanged. Recall
\be
\tQ_{\S_{0^+}}^{\cC} = \int_{\S_{0^+}} \(\q -dC\),
\ee
which can be written as
\be
\tQ_{\S_{0^+}}^{\cC} = \int_{\S_{0^+}} dx d^{D-2}y \sqrt{\g}\, \td\q^t -\int_{\pa\S_{0^+}} C,
\ee
where $\g$ is the determinant of the induced metric on the HRT surface. Since $X^{h_R}$ does not affect $\pa\S_{0^+}$, we find
\ba
\cL_{X^{h_R}} \tQ_{\S_{0^+}}^{\cC} &= \int_{\S_{0^+}} dx d^{D-2}y \sqrt{\g}\, \cL_{X^{h_R}} \td\q^t\\
&= -2\int_{\S_{0^+}} dx d^{D-2}y \sqrt{\g}\, \fr{\pa \cL}{\pa R_{\m\r\n t}} \cL_{X^{h_R}} \D g_{\m\n;\r}\\
&= -2\int_{\S_{0^+}} dx d^{D-2}y \sqrt{\g}\, \fr{\pa \cL}{\pa R_{xtxt}} \cL_{X^{h_R}} \D g_{xx,t}\\
&= -2\int_{\S_{0^+}} dx d^{D-2}y \sqrt{\g}\, \fr{\pa \cL}{\pa R_{xtxt}} \D (4\pi\d(x))\\
&=0.\la{lxhrq}
\ea

Using the point-by-point version of Lemma~\ref{lem1} from Appendix~\re{app:lorentzian}, if follows that on such solutions
the generator $g$ is equal (up to an additive constant) to
\be
\td g= \lt.\( X^{h_R} \cdot \tQ_{\S_{0^+}}^{\cC} \)\rt|_{\tP}.
\ee
It is straightforward to verify
\ba
X^{h_R} \cdot \tQ_{\S_{0^+}}^{\cC} &= -2\int_{\S_{0^+}} dx d^{D-2}y \sqrt{\g}\, \fr{\pa \cL}{\pa R_{xtxt}} X^{h_R} \cdot \D g_{xx,t}\\
&= -2\int_{\S_{0^+}} dx d^{D-2}y \sqrt{\g}\, \fr{\pa \cL}{\pa R_{xtxt}} 4\pi\d(x)\\
&= -8\pi \int_{\g} d^{D-2}y \sqrt{\g}\, \fr{\pa \cL}{\pa R_{xtxt}}.\la{xhrq}
\ea
To see that this is exactly the Wald entropy, recall
\be
S_{\Wald} = -2\pi \int_{\g} d^{D-2}y \sqrt{\g}\, \fr{\pa \cL}{\pa R_{\m\r\n\s}} \e_{\m\r} \e_{\n\s},
\ee
where $\e$ is the 2-dimensional Levi-Civita tensor with $\e_{tx}=-\e_{xt}=1$. Thus
\be
S_{\Wald} = -8\pi \int_{\g} d^{D-2}y \sqrt{\g}\, \fr{\pa \cL}{\pa R_{xtxt}},
\ee
in precise agreement with \er{xhrq}. Therefore, the Wald entropy generates the BCP kink transform in Killing horizon cases.

\subsection{General theories of gravity}
We now extend the argument above to general theories of gravity described by a diffeomorphism invariant Lagrangian containing arbitrary matter fields, as studied in Ref.~\cite{Iyer:1994ys}. In particular, derivatives of Riemann tensors are allowed in the Lagrangian.

Lemma 3.1 of Ref.~\cite{Iyer:1994ys} showed that in such theories, $\q$ can be chosen to have the form
\be
\q = 2 E_R^{\m\n\r} \na_\r \D g_{\m\n} +\q'
\ee
where $\q'$ is a linear combination of $\D g_{\m\n}$, $\D \na_{(\m_1} \cd \na_{m_i)} R_{\m\r\n\s}$, and $\D \na_{(\m_1} \cd \na_{m_i)} \y$ (where $\y$ is any matter field); in other words, the $\D$'s in $\q'$ are always to the left of derivatives of the dynamical fields. Moreover, $E_R^{\m\n\r}$ is defined as
\be
(E_R^{\m\n\r})_{\s_2\cd \s_D} = E_R^{\s\m\n\r} \ve_{\s\s_2\cd \s_D}
\ee
where $E_R^{\s\m\n\r}$ is the equation of motion that would be obtained for $R_{\s\m\n\r}$ if viewed as an independent field.

Using the same Gaussian normal coordinates as in the previous subsection, we again find that $X^{h_R}$ changes $g_{xx,t}$ according to \er{changekxx} while keeping everything else unchanged. In particular, it keeps $\q'$ unchanged. Thus, we find the generalization of \er{lxhrq}:
\ba
\cL_{X^{h_R}} \tQ_{\S_{0^+}}^{\cC} &= \int_{\S_{0^+}} \cL_{X^{h_R}} \(\q -dC\)\\
&= \int_{\S_{0^+}} dx d^{D-2}y \sqrt{\g}\, \cL_{X^{h_R}} \(2 E_R^{t \m\n\r} \na_\r \D g_{\m\n}\)\\
&= 2\int_{\S_{0^+}} dx d^{D-2}y \sqrt{\g}\, E_R^{txxt} \D (4\pi\d(x))\\
&=0.
\ea
Moreover, we find the generalization of \er{xhrq}:
\ba
X^{h_R} \cdot \tQ_{\S_{0^+}}^{\cC} &= 2\int_{\S_{0^+}} dx d^{D-2}y \sqrt{\g}\, E_R^{txxt} X^{h_R} \cdot \D g_{xx,t}\\
&= 2\int_{\S_{0^+}} dx d^{D-2}y \sqrt{\g}\, E_R^{txxt} 4\pi\d(x)\\
&= 8\pi \int_{\g} d^{D-2}y \sqrt{\g}\, E_R^{txxt}.\la{xhrqgen}
\ea

To see that this is exactly the Wald entropy, recall the definition of Ref.~\cite{Iyer:1994ys}:
\be
S_{\Wald} = 2\pi \int_{\g} X^{\n\r} \e_{\n\r}, \qu
(X^{\n\r})_{\s_3\cd \s_D} = -E_R^{\s\m\n\r} \ve_{\s\m \s_3\cd \s_D},
\ee
which becomes in our Gaussian normal coordinates
\be
S_{\Wald} = -2\pi \int_{\g} d^{D-2}y \sqrt{\g}\, 2E_R^{tx \n\r} \e_{\n\r} = 8\pi \int_{\g} d^{D-2}y \sqrt{\g}\, E_R^{txxt}.
\ee
This precisely agrees with \er{xhrqgen}. Again using the point-by-point version of the Lemma~\ref{lem1} from Appendix~\re{app:lorentzian}, we conclude that the Wald entropy generates the BCP kink transform in Killing horizon cases in general theories of gravity.

\addcontentsline{toc}{section}{References}
\bibliographystyle{JHEP}
\bibliography{references}

\end{document}